\documentclass[journal,10pt]{IEEEtran}

\IEEEquantizetextheight{c}
\IEEEsettextwidth{14mm}{14mm}
\IEEEsetsidemargin{c}{0mm}

\usepackage[utf8]{inputenc}
\usepackage[T1]{fontenc}
\usepackage{url}
\usepackage{ifthen}
\usepackage[cmex10]{amsmath} 
\usepackage{amssymb}
\usepackage{algorithm}
\usepackage{algpseudocode}
\usepackage{graphicx}
\usepackage{color}
\usepackage{xcolor}
\usepackage{cite}
\usepackage{subcaption}
\usepackage{comment}
\usepackage{mathtools}
\usepackage{tabulary}
\usepackage[nohyperlinks]{acronym}
\usepackage{upgreek}
\usepackage{placeins}
\usepackage[super]{nth}
\usepackage{hyperref}
\usepackage{array}
\usepackage{soul}
\usepackage{dsfont}
\usepackage{cancel}
\usepackage{multirow}
\usepackage{bm}
\usepackage{bbm}
\usepackage{amsthm}
\usepackage[acronym]{glossaries}
\usepackage{balance}
\usepackage{mathrsfs} 

\glsdisablehyper

\newtheorem{proposition}{\bf Proposition}
\newtheorem{lemma}{\bf Lemma}

\newtheorem{definition}{\bf Definition}
\newtheorem{remark}{Remark}

\newcommand{\tightdots}{%
  \mathinner{\ldotp\mkern-2mu\ldotp\mkern-2mu\ldotp}%
}

\newacronym{RIS}{RIS}{Reconfigurable Intelligent Surface}
\newacronym{NL}{NL}{NonLinear}
\newacronym{SIM}{SIM}{Stacked Intelligent Metasurfaces}
\newacronym{DNN}{DNN}{Deep artificial Neural Network}
\newacronym{NN}{NN}{neural network}
\newacronym{EI}{EI}{Edge Inference}
\newacronym{EM}{EM}{ElectroMagnetic}
\newacronym{TOC}{GOC}{Goal-Oriented Communications}
\newacronym{TX}{TX}{Transmitter}
\newacronym{RX}{RX}{Receiver}
\newacronym{UE}{UE}{User Equipment}
\newacronym{BS}{BS}{Base Station}
\newacronym{LoS}{LoS}{Line-of-Sight}
\newacronym{NLoS}{NLoS}{Non Line-of-Sight}
\newacronym{CNN}{CNN}{Convolutional Neural Network}
\newacronym{FFNN}{FFNN}{Feed-Forward Neural Network}
\newacronym{ReLU}{ReLU}{Rectified Linear Unit}
\newacronym{JSCC}{JSCC}{Joint Source Channel Coding}
\newacronym{E2E}{E2E}{End-to-End}
\newacronym{D2D}{D2D}{Device-to-Device}
\newacronym{DeepSC}{DeepSC}{Deep Semantic Communications}
\newacronym{SISO}{SISO}{Single Input Single Output}
\newacronym{MISO}{MISO}{Multiple Input Single Output}
\newacronym{MIMO}{MIMO}{Multiple-Input Multiple-Output}
\newacronym{OAC}{OAC}{Over-the-Air Computation}
\newacronym{MSE}{MSE}{Mean Squared Error}
\newacronym{SOTA}{SotA}{State of the Art}
\newacronym{SGD}{SGD}{Stochastic Gradient Descent}
\newacronym{IoT}{IoT}{Internet of Things}
\newacronym{AE}{AE}{Auto-Encoder}
\newacronym{PDF}{PDF}{Probability Density Function}
\newacronym{AWGN}{AWGN}{Additive White Gaussian Noise}
\newacronym{CSI}{CSI}{Channel State Information}
\newacronym{CFR}{CFR}{Channel Frequency Response}
\newacronym{ISAC}{ISAC}{Integrated Sensing and Communications}
\newacronym{iid}{i.i.d.}{independent and identically distributed}
\newacronym{MINN}{MINN}{Metasurfaces-Integrated Neural Network}
\newacronym{MS}{MS}{MetaSurface}
\newacronym{CE}{CE}{Cross Entropy}
\newacronym{HRIS}{HRIS}{Hybrid RIS}
\newacronym{RF}{RF}{Radio-Frequency}
\newacronym{VAE}{VAE}{Variational Auto-Encoder}
\newacronym{PSK}{PSK}{Phase Shift Keying}
\newacronym{SVD}{SVD}{Singular Value Decomposition}
\newacronym{MLP}{MLP}{Multi Layer Perceptron}
\newacronym{SNR}{SNR}{Signal-to-Noise Ratio}
\newacronym{MI}{MI}{Mutual Information}
\newacronym{MEC}{MEC}{Multi-access Edge Computing}
\newacronym{WMMSE}{WMMSE}{Weighted Minimum Mean Square Error}
\newacronym{ANN}{ANN}{Artificial Neural Network}
\newacronym{SLFN}{SLFN}{Single hidden Layer Feedforward Network}
\newacronym{ELM}{ELM}{Extreme Learning Machine}
\newacronym{LTI}{LTI}{Linear Time Invariant}
\newacronym{OTA}{OTA}{Over-The-Air}
\newacronym{AM}{AM}{Amplitude Modulation}
\newacronym{LS}{LS}{Least-Squares}
\newacronym{DC}{DC}{Direct Current}
\newacronym{RFC}{RFC}{Radio-Frequency Chain}
\newacronym{ML}{ML}{Machine Learning}
\newacronym{XL}{XL}{eXtremely Large}
\newacronym{PGD}{PGD}{Projected Gradient Descent}
\newacronym{WBCD}{WBCD}{Wisconsin Breast Cancer Dataset}
\newacronym{DMA}{DMA}{Dynamic Metasurface Antenna}
\newacronym{CMS}{CMS}{Cascaded MetaSurfaces}
\newacronym{GO}{GO}{Genetic Optimization}
\newacronym{DFT}{DFT}{Discrete Fourier Transform}
\newacronym{UPA}{UPA}{Uniform Planar Array}
\newacronym{ZF}{ZF}{Zero Forcing}

\title{Extreme Learning Machines in the Wave Domain\\ via XL MIMO with Stacked Intelligent Metasurfaces}
\title{XL MIMO Extreme Learning Machines: Edge Inference with Stacked Heterogeneous Metasurfaces}
\title{XL MIMO Extreme Learning Machines in the Wave Domain: Inference with Stacked Metasurfaces}
\title{XL MIMO Extreme Learning Machines:\\ Over-the-Air Inference with Stacked Metasurfaces}
\title{XL MIMO Extreme Learning Machines:\\ Wave-Domain Inference with Stacked Metasurfaces}

\author{Kyriakos Stylianopoulos,~\IEEEmembership{Graduate~Student~Member,~IEEE}, Mattia Fabiani,~\IEEEmembership{Graduate~Student~Member,~IEEE},\\ Davide Dardari,~\IEEEmembership{Fellow,~IEEE}, and George C. Alexandropoulos,~\IEEEmembership{Senior~Member,~IEEE}\vspace{-0.2cm} 
\thanks{A preliminary version of this paper has been presented at IEEE SPAWC, Athens, Greece, September 2026~\cite{Our_SPAWC_paper}.}
\thanks{K. Stylianopoulos and G. C. Alexandropoulos are with the Department of Informatics and Telecommunications, National and Kapodistrian University of Athens, 16122 Athens, Greece (e-mails: \{kstylianop, alexandg\}@di.uoa.gr).}
\thanks{M. Fabiani and D. Dardari are with the Dipartimento di Ingegneria dell’Energia Elettrica e dell’Informazione ``Guglielmo Marconi'' (DEI), WiLAB-CNIT, University of Bologna, Cesena Campus, Cesena (FC), Italy, (e-mails: \{mattia.fabiani5, davide.dardari\}@unibo.it).}
\thanks{This work has been supported by the SNS JU projects 6G-DISAC and TIMES under the EU's Horizon Europe research and innovation program under grant agreement numbers 101139130 and 101096307, respectively.
}}
  
\begin{document}
\maketitle

\begin{abstract}
The recently envisioned goal-oriented communications paradigm requires machine learning inference to be performed directly on wirelessly transferred data, by splitting neural network computations between the transmitter, the Receiver (RX), and the wireless channel.
Such approaches aim to offload large portions of the computation onto the channel itself, which is often controlled by reconfigurable \glspl{MS}. Typical MS-parametrized systems, however, enable only limited computation offloading in the wave domain, usually corresponding to a single neural network layer, due to adverse channel fading effects and the almost linear operation of the channel and conventional MS models. In this paper, we demonstrate that an eXtremely Large (XL) Multiple-Input Multiple-Output (MIMO) system equipped with Stacked Intelligent MSs (SIM) at the RX side, together with a single antenna and its accompanying radio-frequency chain, may act as a two-layer neural network operating exclusively over the air. The first (uncontrollable) network layer corresponds to the transformations offered by wireless fading, while the second (controllable) layer is implemented via SIM operations. In particular, only the first, front-end \gls{MS} layer needs to exhibit nonlinear response to serve as an activation function.
Leveraging the Extreme Learning Machine (ELM) framework, the proposed XL MIMO system architecture is shown to be a universal approximator, and the optimal linear-SIM-based combining is derived in closed form, allowing for fast training even and partial signal observability.
Numerical evaluations across diverse datasets showcase that the proposed XL-MIMO-ELM framework achieves classification accuracy comparable to idealized digital models, thereby, proving the viability of low complexity, wave-domain neural computations.

\end{abstract}


\section{Introduction}

Future wireless networks are expected to support increasingly demanding local processing of sensory observations under stringent communication and computational constraints.
Conventional architectures, which separate the acquisition, transmission, and processing of data, become inefficient in such settings, as the \gls{TX} typically extracts and encodes features without direct knowledge of the downstream inference task, while the \gls{RX} first reconstructs the transmitted data before performing the desired inference.
This separation can lead to unnecessary data transmission, increased processing complexity, and redundant computations at both ends of the communication link.
To address these limitations, future wireless systems will leverage \gls{EI}, where the \gls{TX} and \gls{RX} are jointly trained as an end-to-end \gls{ML} model tailored to a target inference task~\cite{GO_review2}.
Instead of transmitting a representation designed for faithful data reconstruction, \gls{EI} enables the \gls{TX} to extract task-relevant features while simultaneously encoding them for transmission, allowing the \gls{RX} to directly infer the target values without reconstructing the original input~\cite{GJZ24_SIM_TOC, Stylianopoulos_GO}.
By aligning communication and inference within a single end-to-end architecture, \gls{EI} reduces both the amount of information that needs to be conveyed over the wireless channel and the computational resources required for downstream processing.

To further alleviate the computational costs associated with end-to-end wireless intelligence, the emerging \gls{OTA} computing paradigm exploits the intrinsic superposition property of wireless propagation to perform computations directly in the wave domain~\cite{OTA_review}. To this end, rather than receiving individual signals and processing them digitally, multiple transmitted waveforms can naturally combine over the wireless channel, thereby implementing part of the desired computation through their physical propagation.
This property has motivated a growing body of work on wireless physical \glspl{DNN}, which exploit waveform superposition together with transceiver processing for multi-antenna systems to realize computational transformations entirely \gls{OTA}~\cite{reus2023airfc,yang2023over}.

In parallel, recent advancements in metamaterial sciences have resulted in MetaSurface (MS) structures capable to perform \gls{DNN}-equivalent operations by shaping the electromagnetic signals being diffracted through them~\cite{XYN18_Diffractive_DNN,Momeni2022_wave_DL_computing}, giving rise to the concept of wave-domain signal processing~\cite{9475155,DarTorPasDec:J26}.
Subsequent approaches integrated \gls{MS} structures~\cite{AlexandropoulosRIS} into wireless systems to perform \gls{OTA} computations enabled by wave-domain signal processing, thus establishing the propagation environment as a trainable hidden module within an end-to-end \gls{ML} model~\cite{Stylianopoulos_GO,GJZ24_SIM_TOC,Styl:25_Mag}.
Hardware platforms based on \glspl{MS} have been proposed to emulate \gls{DNN} layers for \gls{OTA} inference, with their physical parameters trained jointly through backpropagation~\cite{Stylianopoulos_GO,GJZ24_SIM_TOC}, or to approximate digitally trained \gls{DNN} weight matrices directly in the wave domain~\cite{Gunduz_Layer_Approximation,Pandolfo_SIM_semantic_alignment}.
Relevant architectures have exploited programmable multipath propagation and frequency-selective channels to emulate one- and two-dimensional convolutional operations~\cite{sanchez2022airnn,zhang2024radio,hua2026cnns}. However, despite this progress, the computational transformations realized by many existing physical \gls{MS} architectures remain predominantly linear.
Consequently, the computational savings obtained from moving linear operations to the physical domain may be insufficient to compensate for the additional overhead associated with channel acquisition, analog hardware configuration, and overall system control.

Introducing \gls{NL} functionality into the wireless medium is therefore essential for realizing deeper and more expressive physical neural architectures, as well as for providing stronger theoretical guarantees, such as universal approximation~\cite{Cyb:89}) and faster training procedures. In this regard, several recent approaches have explored this direction by constructing more elaborate propagation and hardware configurations, under the wider umbrella of analog computing.
For example, neural depth has been introduced through iterative transmissions over distributed reconfigurable channels~\cite{liu2025over}, while inter-element coupling under rich-scattering conditions has been very recently exploited to approximate \gls{NL} functionality without explicitly introducing conventional activation layers~\cite{hammami2026expressivity}.
Other architectures employ multi-hop relay networks in which \gls{NL} power amplifiers provide the activation functions required to build deeper neural computations~\cite{bergel2024non,hua2026multi}.
However, these approaches generally rely on dedicated multi-stage physical topologies, repeated signal transmissions, or specialized \gls{NL} hardware.
These observations motivate physical architectures that can introduce \gls{NL} representation capabilities, while retaining the low complexity and rapid reconfigurability of wave-domain computations.

More closely related to this paper, a wave-domain computing implementation for \gls{XL} \gls{MIMO} systems has been recently proposed in~\cite{Stylianopoulos_MIMO_ELM}, the architecture of which operates according to the \gls{ELM} framework~\cite{HuaZhuSie:06}, enabling the hidden layer of the \gls{SLFN} to be executed entirely \gls{OTA}.
The activation function and the weights of the output layer have been assumed to be realized through analog combining hardware, exploiting the nonlinearities introduced by power amplifiers operating in their saturation region~\cite{Rapp_model}. 
Since only the output-layer parameters need to be trained, the architecture enables rapid training and efficient reconfiguration when the wireless channel changes, while preserving the universal function approximation capability of \glspl{ELM}.
Despite providing a way for \glspl{SLFN} \gls{OTA} realization in traditional wireless systems, the original formulation is subject to practical limitations and scalability constraints, particularly due to the real-valued signal representation and the hardware complexity associated with the inclusion of multiple \gls{NL} analog components.
These limitations motivate the development of more practical and scalable physical solutions capable of retaining the computational advantages of the \gls{ELM}-based \gls{EI} approach, while better matching the constraints of wireless hardware components.

In this paper, capitalizing on the complex-domain \gls{ELM} framework~\cite{HuaLiCheSie:08} and building upon recent \gls{NL} metamaterial advancements~\cite{Styl:25_Mag,FabTorDar:25,Omr:25}, we introduce a novel \gls{XL}-\gls{MIMO}-\gls{ELM} system architecture incorporating an \gls{RX} structure comprising heterogeneous \gls{SIM}~\cite{AXN23_SIM} followed by a single antenna and its associated~\gls{RFC}. 
The first, front-end \gls{MS} layer at the \gls{RX} interfacing with the \gls{MIMO} channel is composed of metamaterials exhibiting identical fixed \gls{NL} responses and implements the \gls{ELM} activation function, while the subsequent linear \gls{MS} layers, forming a conventional linear \gls{SIM}~\cite{SIM_survey}, realize trainable wave-domain combining, effectively approximating the digital \gls{ELM} output weights. 
The proposed \gls{XL}-\gls{MIMO}-\gls{ELM} framework enables fast (gradient-free) training and minimizes digital processing at the \gls{RX} side, while significantly reducing the hardware complexity with respect to the system architecture in~\cite{Stylianopoulos_MIMO_ELM}.
In particular, the contributions of this paper are summarized as follows:
\begin{itemize}
    \item We introduce a novel \gls{XL} \gls{MIMO} system architecture incorporating a \gls{SIM} with heterogeneous \gls{NL}/linear \gls{MS} layers at the \gls{RX} side, highlighting the required analog components for activation-function-like \gls{NL} front-end operation.
    \item We present a novel mathematical formulation showcasing that the proposed \gls{XL} \gls{MIMO} system operates as an \gls{ELM} on the complex domain, thereby positioning our \gls{XL}-\gls{MIMO}-\gls{ELM} framework as an \gls{EI} model realized exclusively \gls{OTA} via dedicated wave-domain signal processing.
    \item We describe the training procedure for our \gls{XL}-\gls{MIMO}-\gls{ELM}, and design a reduced overhead alternative tailored to \glspl{RX} equipped with a single antenna and respective \gls{RFC}.
    \item A theoretical proof of the universal approximation capabilities of our \gls{XL}-\gls{MIMO}-\gls{ELM} system is presented, following mathematical steps similar to conventional neural networks. 
    \item We present a closed-form approach for realizing single-antenna analog combining entirely in the wave domain via configuration of the last linear \gls{MS} layer of the \gls{RX} \gls{SIM}.
\end{itemize}
Overall, this paper presents the first fully analog neural network accommodating: \textit{i)} an analog multi-antenna \gls{TX}; \textit{ii)} Ricean \gls{MIMO} channel conditions; and \textit{iii)} \gls{MS}-based activation and combining realized in the wave domain, while simultaneously providing closed-form training and theoretical guarantees. The proposed \gls{XL}-\gls{MIMO}-\gls{ELM} architecture operates with almost zero latency, while also benefiting from low power consumption and reduced hardware complexity. The presented numerical results demonstrate the efficiency of the proposed  framework in numerous binary classification examples, approaching the performance of fully digital models at the \gls{XL} \gls{MIMO} regime.

This paper extends our recent conference version~\cite{Our_SPAWC_paper} in the following aspects: \textit{i)}~the \gls{SIM} propagation is modeled through physically grounded Rayleigh--Sommerfeld coefficients~\cite{AXN23_SIM} instead of pseudo-random matrices, and the \gls{RX} architecture (front-end \gls{NL} \gls{MS} followed by a linear \gls{SIM} and single-antenna reception) is described in detail; \textit{ii)}~the universal approximation property, only sketched in~\cite{Our_SPAWC_paper}, is fully proven (Proposition~\ref{prop:universal-approximation}) and complemented by an analysis of the effect of reception noise (Proposition~\ref{prop:noise-effects}); \textit{iii)}~the element-by-element measurement of the activation matrix is generalized to a sketching-based procedure allowing amount of measurements suited to a single RFC, with exact recovery and monotonicity guarantees (Propositions~\ref{prop:equal-weights-for-Nr} and~\ref{prop:monotonicity}); \textit{iv)}~the iterative linear \gls{SIM} training of~\cite{Our_SPAWC_paper} is replaced by a closed-form, gradient-free configuration of its last \gls{MS} layer that implements the trained combining weights in the wave domain (Propositions~\ref{prop:closed-form-MS} and~\ref{prop:sim-full-rank}); and \textit{v)}~the numerical evaluation is extended to the impact of the sketching ratio, to reception noise during both training and inference, and to Ricean and power level analyses over additional dataset and parameter combinations.

The remainder of this paper is organized as follows:
Section~\ref{sec:system-model} provides the complete \gls{XL}-\gls{MIMO} system model.
Section~\ref{sec:methodology} presents the formulation of the system as an \gls{ELM} and provides theoretical guarantees.
The complete training and inference procedures are detailed in Section~\ref{sec:training-deployment}, while a numerical evaluation is included in Section~\ref{sec:results}.
Finally, Section~\ref{sec:conclusion} concludes the paper.

\textbf{Notation:} Vectors, matrices, and sets are expressed in lowercase bold, uppercase bold, and uppercase calligraphic typefaces, respectively.
$\mathbf{X}^\top$, $\mathbf{X}^{\rm H}$, and $\mathbf{X}^\dagger$ denote the transpose, conjugate transpose, and pseudo-inverse of $\mathbf{X}$.
The notation $[\mathbf{a}, \dots, \mathbf{z}]$ indicates a matrix (or vector, depending on the dimensionality of $\mathbf{a}, \dots, \mathbf{z}$) where $\mathbf{a}, \dots, \mathbf{z}$ are stacked horizontally. Likewise, the notation $[\mathbf{a}; \dots; \mathbf{z}]$ implies vertical stacking. 
$[\mathbf{x}]_i$ is used to denote the $i$th element of $\mathbf{x}$.
$|\mathcal{X}|$ represents the cardinality of the set $\mathcal{X}$, $||\mathbf{X}||_{\rm F}$ denotes the Frobenius norm of $\mathbf{X}$, ${\rm diag(\mathbf{x})}$ creates a square matrix with the elements of $\mathbf{x}$ placed along its main diagonal.
$\{\cdot\}$ expresses a set or collection,
${\rm sign}(x)$ returns the sign of $x$,
while ${\rm Tr}(\mathbf{X})$ and ${\rm det}(\mathbf{X})$ express the trace and the determinant of $\mathbf{X}$.
Finally, $\mathbf{X} \sim \mathcal{CN}(\mathbf{0}, \sigma^2\mathbf{I})$ signifies that $\mathbf{X}$ follows the complex standard Gaussian distribution with variance $\sigma^2$, $\jmath\triangleq\sqrt{-1}$ is the imaginary unit, and $z \in \mathbb{C}$ which can be expressed as $z = \mathfrak{Re}(z) + \jmath \mathfrak{Im}(z) = |z| \exp(\jmath~ {\rm arg}\{z\})$.

\section{XL MIMO System Model} \label{sec:system-model}

Consider a narrowband \gls{XL} \gls{MIMO} system with an $N_{\rm t}$-antenna \gls{TX} and an \gls{RX} comprising a \gls{SIM}, with a front-end \gls{MS} of $N_{\rm r}$ diffractive elements, followed by a single antenna and its accompanying \gls{RFC}. The methodology presented in this paper also holds for multiple synchronized \glspl{TX} (e.g., sensor devices relaying raw measurements~\cite{6G-DISAC-Architecture}) comprising a virtual sparse antenna; however, we leave such systemic elaborations for future work for the sake of focus.
Instead of performing conventional wireless communications, the system is trained end-to-end to act as an \gls{OTA} function approximator~\cite{Stylianopoulos_MIMO_ELM}.
In particular, given a prior dataset $\mathcal{D} \triangleq \{(\mathbf{x}^{(i)}, z^{(i)})\}_{i=1}^D$ of $D$ input-target pairs, the \gls{XL} \gls{MIMO} system is designed to approximate the $\mathbf{x} \to z$ mapping so that the \gls{TX} observes the input data $\mathbf{x}$ (not necessarily belonging to $\mathcal{D}$) and the \gls{RX} infers the (unobserved) target feature $z$ of $\mathbf{x}$. As shown in Fig.~\ref{fig:system-model}, which will be further elaborated on during the discussion of the methodology, the proposed \gls{XL} \gls{MIMO} system is thus intended to perform \gls{EI}, with all computational processing performed exclusively in the wave domain, i.e., \gls{OTA}. 
\begin{figure}
    \centering
    \includegraphics[width=1.0\linewidth]{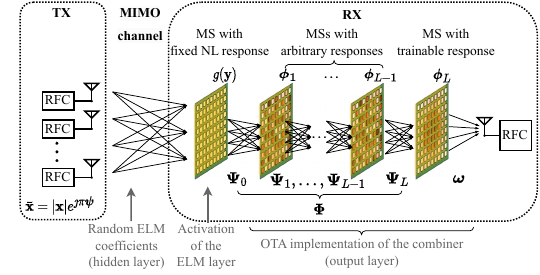}
    \caption{The proposed \gls{XL} \gls{MIMO} system architecture for implementing the devised \gls{XL}-\gls{MIMO}-\gls{ELM} framework using linear and \gls{NL} \glspl{MS}. The responses of the channel and \glspl{MS} are used as components of the \gls{ELM} algorithm to realize \gls{OTA} inference.}
    \label{fig:system-model}
\end{figure}

\subsection{Signal and Channel Models}
We assume that $\mathbf{x}^{(i)} \in (0,1)^{N_{\rm t}}$ and, without loss of generality, $z^{(i)} \in [-1,+1]$, that is, the dimension of the data observations is equal to the number of \gls{TX} antennas\footnote{It is noted that the \gls{TX} may incorporate a trainable feature extraction module to reduce the dimensionality of $\mathbf{x}$~\cite{Stylianopoulos_GO, GJZ24_SIM_TOC}, and thus the number of TX antennas. However, in this paper, we do not consider such multi-stage processes to the purpose of providing a clear \gls{OTA} computing framework for \gls{EI}, and to simplify the overall \gls{XL} \gls{MIMO} system by requiring a single \gls{RX} \gls{RFC}. Additional digital pre- or post-processing components may be added as is typical in \gls{EI} systems, towards heterogeneous architectures~\cite{Stylianopoulos_GO}; this is left for future investigation.} and the dataset is used for real-valued estimation. For reasons that will be clarified later, the data symbols are subject to \gls{AM}; accordingly, each element of the transmitted signal vector $\bar{\mathbf{x}} \in \mathbb{C}^{N_{\rm t} \times 1}$ is defined $\forall i=1,\dots,N_{\rm t}$ as follows:
\begin{equation}\label{eq:transmit-signal}
    [\bar{\mathbf{x}}]_i \triangleq [\mathbf{x}]_i \exp(\jmath \pi  [\bm \psi]_i),
\end{equation}
where the $N_{\rm t}$ elements of $\bm \psi$ may be chosen arbitrarily.
The baseband representation of the impinging signal at the front-end \gls{MS} layer of the \gls{SIM} at the \gls{RX}, which is composed of $N_{\rm r}$ metamaterial elements, can be mathematically expressed as:
\begin{equation}\label{eq:received-signal}
    \mathbf{y} \triangleq \mathbf{H} \bar{\mathbf{x}} \in \mathbb{C}^{N_{\rm r} \times 1},
\end{equation}
where $\mathbf{H} \in \mathbb{C}^{N_{\rm r} \times N_{\rm t}}$ represents the \gls{XL} \gls{MIMO} channel gain matrix.
For the subsequent theoretical analysis, we consider the case where $\mathbf{H}$ remains quasi-static for the duration of the training process and follows the Ricean fading distribution~\cite{6184250}:
\begin{equation}\label{eq:ricean}
   \mathbf{H} = \sqrt{P_L} \left( \sqrt{\frac{K}{1+K}} \mathbf{H}_{\rm LoS} + \sqrt{\frac{1}{1+K}} \mathbf{H}_{\rm NLoS} \right).
\end{equation}
where $P_L$ represents the pathloss. Assuming TX and RX \glspl{UPA} with $N_{\rm t}=N_{{\rm t},x}N_{{\rm t},y}$ and $N_{\rm r}=N_{{\rm r},x}N_{{\rm r},y}$ elements, respectively, the rank-one \gls{LoS} component is given by:
\begin{align}
    \mathbf{H}_{\rm LoS}
    \triangleq
    \mathbf{a}_{\rm r}\!\left(\theta^{\rm AoA},\phi^{\rm AoA}\right)
    \mathbf{a}_{\rm t}^{\rm H}\!\left(\theta^{\rm AoD},\phi^{\rm AoD}\right),
\end{align}
where $(\theta^{\rm AoD},\phi^{\rm AoD})$ and $(\theta^{\rm AoA},\phi^{\rm AoA})$ denote the elevation and azimuth angles of departure and arrival, respectively. Defining the string $q\in\{{\rm t},{\rm r}\}$, the unit-norm steering vector of the corresponding \gls{UPA} is defined as follows:
\begin{align}
    \big[\mathbf{a}_{q}(\theta,\phi)\big]_{m+nN_{q,x}+1}
    \triangleq
    \frac{1}{\sqrt{N_q}}
    \exp\!\Big(&
        \jmath\frac{2\pi}{\lambda}
        \big[
            m d_{q,x}\sin\theta\cos\phi
        \notag\\[-1mm]
            &+n d_{q,y}\sin\theta\sin\phi
        \big]
    \Big).
\end{align}
for $m=0,\ldots,N_{q,x}-1$ and $n=0,\ldots,N_{q,y}-1$, where $d_{q,x}$ and $d_{q,y}$ are the inter-element spacings along the two array axes and $\lambda$ is the carrier wavelength.
Additionally, in~\eqref{eq:ricean}, $\mathbf{H}_{\rm NLoS} \sim \mathcal{CN}(\mathbf{0}, 1/\sqrt{N_{\rm t} N_{\rm r}}\mathbf{I})$, and $K$ represents the Ricean factor that controls the dominance of either component.

\begin{figure}
    \centering
    \includegraphics[width=0.7\linewidth]{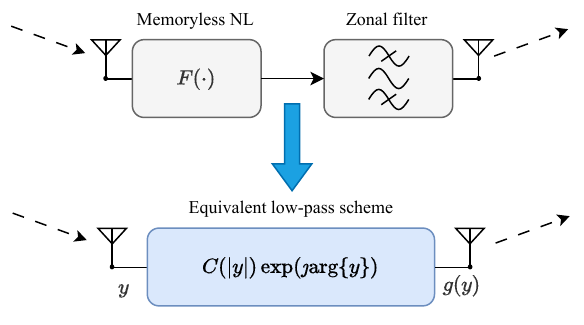}
    \caption{Element-wise bandpass memoryless NL transformation (top) and its corresponding equivalent low-pass scheme (bottom).}
    \label{fig:NL}
\end{figure}

\subsection{The Proposed RX Architecture}\label{sec:NL_SIM}
As shown in Fig.~\ref{fig:system-model}, the proposed RX architecture comprises a \gls{SIM} followed by a single RX antenna and its accompanying \gls{RFC}. The SIM further consists of $L+1$ diffractive \gls{MS} layers indexed with $l=0,\dots,L$: the front-end \gls{MS} with $l=0$ implements a NL transformation on the transmitted signal, while the remaining $L$ \glspl{MS} with $l=1,\ldots,L$ form a typical linear \gls{SIM}~\cite{AXN23_SIM}. In the sequel, we describe each of these components.

\subsubsection{Front-End Nonlinear Metasurface}


Each of the $N_{\rm r}$ elements of the front-end \gls{MS} applies a memoryless \gls{NL} transformation followed by a zonal filter, as shown in Fig.~\ref{fig:NL}, aiming to remove higher-order harmonics; this filtering task can be accomplished, for instance, directly by the antenna/MS element. Denoting with $F(\cdot)$ the bandpass response of the generic \gls{NL} transformation, the fundamental component produced by each \gls{NL} element of this \gls{MS}, under a narrowband excitation, is given by the \gls{AM}/\gls{AM} characteristic. It can be obtained by extracting the first cosine Fourier coefficient, and it is expressed as~\cite{Maas2003}:
\begin{equation} \label{eq:C}
C(v)\triangleq\frac{2}{\pi} \int_{0}^{\pi} F(v\cos(\phi)) \cos(\phi) \text{d}\phi \, .
\end{equation}
The resulting equivalent low-pass mapping of the whole \gls{NL} \gls{MS} can be thus derived as $g(\mathbf{y}) \triangleq C(|\mathbf{y}|) \exp(\jmath {\rm arg}\{\mathbf{y}\})$. 
It is worth noticing that every memoryless bandpass \gls{NL} transformation preserves the phase of the input, while transforming its envelope through the \gls{AM}/\gls{AM} characteristic. 

In this paper, we consider an element-wise thresholding \gls{NL} transformation characterized overall by the positive bias vector $\mathbf{b} \in \mathbb{R}^{N_{\rm r} \times 1}_+$, whose entries are drawn from an appropriate distribution; for convenience, it may be centered around the expected receive signal power.
As suggested in~\cite{FabTorDar:25,Abr:26}, in the simplest passive implementation, such a thresholding can be approximated by exploiting the nonlinearity of an anti-parallel diode pair connected in shunt across the transmission line. The different biases can be obtained during the fabrication process by varying the diode's polarization or its characteristic, hence, ensuring low complexity. Higher values of biases can be obtained by resorting to active \gls{NL} circuits as done in \cite{Nin:25,Abb:26}.   
From \eqref{eq:C}, the transform yields the following piecewise mapping for each element $j=1,\dots,N_{\rm r}$ of the front-end \gls{NL} MS layer~\cite{FabTorDar:25}:
\begin{equation}
    C(|[\mathbf{y}]_j|) = \begin{cases} 
    \mathbf{0}, & |[\mathbf{y}]_j| \le [\mathbf{b}]_j \\
    \begin{aligned}
    \frac{1}{\pi} \left( |[\mathbf{y}]_j| \arccos\left(\frac{[\mathbf{b}]_j}{|[\mathbf{y}]_j|}\right)\right. & \\
    \left.- [\mathbf{b}]_j \sqrt{1 - \left(\frac{[\mathbf{b}]_j}{|[\mathbf{y}]_j|}\right)^2}\right) & ,
    \end{aligned} & |[\mathbf{y}]_j| > [\mathbf{b}]_j
    \end{cases}.
    \label{eq: C[|y|]}
\end{equation}

It is noted that, while the latter expression captures the exact \gls{AM}/\gls{AM} characteristic of the \gls{NL} elements, the transcendental terms are computationally demanding for practical optimization, due to their non-differentiability at each  $[\mathbf{b}]_j$. To this end, by approximating the transition for $|[\mathbf{y}]_j| > [\mathbf{b}]_j$ as a continuous quasi-linear slope, the response of the \gls{NL} \gls{MS} layer can be expressed as the well known softplus activation function $\forall j$~\cite{Wiemann2024_softplus}:
\begin{equation}\label{eq:softplus-activation}
    g([\mathbf{y}]_j) \simeq 0.2
    \log\left( 1+e^{5 (|[\mathbf{y}]_j| -[\mathbf{b}]_j)  }\right)
    \exp\left(\jmath {\rm arg}\{[\mathbf{y}]_j\}\right).
\end{equation}

\subsubsection{Stacked Linear Metasurfaces}

The $L$ linear \gls{MS} layers following the \gls{NL} \gls{MS} implementing \eqref{eq:softplus-activation}'s activation function constitute of response-tunable elements which can collectively process $g(\mathbf{y})$ in the wave domain. 
The element-to-element propagation between the $l$-th and ($l+1$)-th \gls{MS} layers ($l=0,\dots,L$) 
is governed by geometric optics due to their dense placement~\cite{AXN23_SIM}.
Given elements $n$ and $n'$ ($n,n'=1,\dots,N_{\rm r}$) with distance $d_{n,n'}$ and area $S_M$ from layers $l$  and $l-1$ ($l=1,\dots,L$) of distance $d_M$, each propagation matrix $\mathbf{\Psi}_l\! \in \! \mathbb{C}^{N_{\rm r} \times N_{\rm r}}$ can be expressed as:
\begin{align}\label{eq:rayleigh-somerfeld}
    [\mathbf{\Psi}_l]_{n,n'} &\triangleq \frac{d_M S_M}{d_{n,n'}^2} 
    \Big( \frac{1}{2\pi d_{n,n'}} - \frac{\jmath}{\lambda} \Big) 
    \exp\left(\frac{\jmath 2\pi}{\lambda} d_{n,n'}\right).
\end{align}
The linear responses of the elements of each $l$-th \gls{MS} layer ($l=1,\dots,L$) are modeled as $[\boldsymbol{\phi}_l]_n \triangleq |\alpha^l_n|\exp(-\jmath \vartheta^l_n)$, where $|\alpha^l_n|\leq 1$ is the controllable amplitude and $\vartheta^l_n$ the controllable phase.
By defining $\mathbf{\Phi}_l \triangleq {\rm diag}(\boldsymbol{\phi}_l)$ for $l=1,\dots,L$, the linear \gls{SIM} response can be mathematically expressed via the matrix:
\begin{equation}\label{eq:sim-propagation}
 \mathbf{\Phi} \triangleq \prod_{l=L}^{1} \mathbf{\Phi}_l \mathbf{\Psi}_l   \in \mathbb{C}^{N_{\rm r} \times N_{\rm r}},
\end{equation}
where $\mathbf{\Psi}_1$ is the propagation between the output of the front-end \gls{NL} \gls{MS} and the input of the first \gls{MS} layer of the linear \gls{SIM}.

\subsubsection{Single-Antenna Reception}

For the considered scalar estimation \gls{EI} problem in this paper, the \gls{RX} needs only a single antenna associated with an \gls{RFC} for digital conversion.
For simplicity, we position the \gls{RFC} along the line passing through the centers of all $L+1$ \gls{MS} layers, having a distance of $d_M$ from the last $L$-th layer, identical to the inter-element distances.
Due to its close proximity, this propagation defined as the row vector $\boldsymbol{\omega} \in \mathbb{C}^{1 \times N_{\rm r}}$ between the $L$-th linear \gls{MS} layer (see Fig.~\ref{fig:system-model}) and the \gls{RFC} can be modeled, similar to~\eqref{eq:rayleigh-somerfeld}, as follows:
\begin{align}\label{eq:rayleigh-somerfeld-vector}
    [\boldsymbol{\omega}]_{n} &\triangleq \frac{d_M S_M}{\tilde{d}_{n}^2} 
    \Big( \frac{1}{2\pi \tilde{d}_{n}} - \frac{\jmath}{\lambda} \Big) 
    \exp\left(\frac{\jmath 2\pi}{\lambda} \tilde{d}_{n}\right) \in \mathbb{C},
\end{align}
where $\tilde{d}_{n}$ denotes the distance between the $n$-th element of the last $L$-th \gls{MS} layer and the single antenna.

Putting all above together, the received signal at the output of the reception \gls{RFC}, i.e., in baseband, can be expressed as:
\begin{align}
    \hat{z} \triangleq \boldsymbol{\omega} \mathbf{\Phi} g(\mathbf{y}) + \tilde{n} \overset{(a)}{=} \mathbf{w}^{\rm H} g(\mathbf{y}) + \tilde{n}, \label{eq:output}
\end{align}
where $\tilde{n} \in \mathcal{CN}(0, \sigma_{\tilde{n}}^2)$ denotes the \gls{AWGN} introduced at the \gls{RX}, and (a) uses the definition $\mathbf{w} \triangleq (\boldsymbol{\omega} \mathbf{\Phi})^{\rm H} \in \mathbb{C}^{N_{\rm r} \times 1}$ for the combining operation in the wave domain (i.e., implemented via the $L$-layer linear SIM), whose response can be controlled to solve computational \gls{EI} tasks, as will be described in the next section.

\section{The Proposed XL-MIMO-ELM Framework}\label{sec:methodology}
In this section, we show that our \gls{XL} \gls{MIMO} system architecture in Fig.~\ref{fig:system-model} can be regarded as a form of an \gls{ELM}. To this end, let us briefly describe the \gls{ELM} model to be used for inference tasks on the data set $\mathcal{D}$. In particular, an \gls{ELM} for scalar-value estimation is defined as an \gls{SLFN} with $\mathtt{N}$ neurons in its hidden layer, which is tasked to perform the following simple operation on the input $\mathbf{x}$. This is mathematically expressed as~\cite{HuaZhuSie:06}:
\begin{equation}\label{eq:elm-theory}
    \hat{\mathtt{z}} \triangleq \sum_{j=1}^{\mathtt{N}} [\mathtt{w}]_j \mathtt{g}(\mathtt{h}_j(\mathbf{x}; \boldsymbol{\theta}_j)) = \mathtt{w}^{\rm H} \mathtt{g}(\mathtt{H}(\mathbf{x}; \boldsymbol{\theta})),   
\end{equation}
where $\mathtt{w}$ is its trainable weight vector of the output layer, $\mathtt{g}(\cdot)$ is an activation function, and $\mathtt{H}(\mathbf{x}; \boldsymbol{\theta}) \triangleq [\mathtt{h}_1(\mathbf{x}; \boldsymbol{\theta}_1), \dots, \mathtt{h}_{\mathtt{N}}(\mathbf{x}; \boldsymbol{\theta}_{\mathtt{N}})]^\top$ includes the kernel functions of the $\mathtt{N}$ neurons of the hidden layer, each parametrized by $\boldsymbol{\theta}_j$; we use the definition $\boldsymbol{\theta} \triangleq \{ \boldsymbol{\theta}_j \}_{j=1}^{\mathtt{N}}$.
Such kernel functions may belong to a wide family of functions, and typically, even linear functions are an appropriate choice.
The important aspect of in~\glspl{ELM} is that $\boldsymbol{\theta}_j$'s are completely random and need not be optimized.
The associated theoretical results~\cite{HuaZhuSie:06} demonstrate that, if $\mathtt{N}$ is arbitrarily large (yet finite), then simply optimizing  $\mathtt{w}$ through an analytical \gls{LS} solution is sufficient to make $\hat{\mathtt{z}}^{(i)}$ approximate $z^{(i)}$ with arbitrary precision $\forall \mathbf{x}^{(i)}$ in $\mathcal{D}$.
In that respect, the \gls{ELM} framework combines random parameters of the hidden layers, which acts as a random feature extraction mechanism, with trainable weights at the output layer, that act as a weighted selection of useful features.
Their aforementioned universal approximation capability can thus be interpreted to mean that if the number of the random feature projections is sufficiently large, each data point has a unique representation in that random space.

\subsection{XL MIMO as an ELM}\label{sec:elm-definition}
\subsubsection{The Idealistic Noise-Free Case}\label{sec:noise_free}
Inspecting the transformations in the previous expressions~\eqref{eq:received-signal} and~\eqref{eq:softplus-activation}, it can be easily concluded that they implement an ELM's random hidden layer with random coefficients
${\bm \theta} \!=\! [ \mathfrak{Re}([\mathbf{H}]_{1,:}),\!\tightdots\!, \mathfrak{Re}([\mathbf{H}]_{N_{\rm r},:}), \mathfrak{Im}([\mathbf{H}]_{1,:}),\!\tightdots\!, \mathfrak{Im}([\mathbf{H}]_{N_{\rm r},:}),\! \mathbf{b}^\top ]^\top$
$\in \mathbb{R}^{2 N_{\rm r} N_{\rm t} + N_{\rm r}}$, and the combining weights $\mathbf{w}$ of~\eqref{eq:output} play the role of ELM's trainable weights $\mathtt{w}$ in~\eqref{eq:elm-theory} of the output layer.
Motivated by this operational equivalence, we leverage the mathematical framework for \gls{ML} inference in~\cite{HuaZhuSie:06, LiHuaSarSun:05, Stylianopoulos_MIMO_ELM}, which accounts for random parameters alongside $\mathtt{w}$, enabling both the training procedure of the ELM framework and its theoretical guarantees to be rigorously described.

To find the optimal weight vector for a given dataset $\mathcal{D}$, let us first define the \gls{ELM} activation matrix $\mathbf{G} \in \mathbb{C}^{D \times N_{\rm r}}$ as the transpose of the activated signals at the aforedescribed front-end \gls{NL} \gls{MS} layer at the RX side of the proposed \gls{XL} \gls{MIMO} system:
\begin{equation}\label{eq:hidden-layer-outputs}
    \mathbf{G} \triangleq [g^{\rm H}(\mathbf{y}^{(1)}) ; \dots ; g^{\rm H}(\mathbf{y}^{(D)})].
\end{equation}
The output of the system for the whole dataset can now be compactly written, following~\eqref{eq:output} and ignoring the \gls{AWGN} term (implying operation at the high-\gls{SNR} regime for the moment), as follows:
\begin{align}
  \hat{\mathbf z} &\triangleq [ \mathbf{w}^{\rm H}\mathbf g(\mathbf y^{(1)});\, \dots \mathbf{w}^{\rm H} \mathbf g(\mathbf y^{(D)}) ] \nonumber \\
  &= [ \mathbf g^{\rm H}(\mathbf y^{(1)})\mathbf{w}; \,  \dots ;\, \mathbf g^{\rm H}(\mathbf{ y}^{(D)}) \mathbf{w} ]^{\ast} = (\mathbf{G}\mathbf{w})^{\ast} \in \mathbb{C}^{D \times 1}.\label{eq:z_hat}
\end{align}
We further denote the vector of target values for the whole dataset as $\mathbf{z} \triangleq [z^1 , \dots , z^D]^\top \in \mathbb{R}^{D \times 1}$. Since the purpose of the inference system is to approximate $\mathbf{z}$, which is a real vector for our considered estimation problems, we assume that $\hat{\mathbf{z}}$ also resides in the real domain (a formal proof of exact approximation will be presented next), therefore $\hat{\mathbf{z}} \cong \hat{\mathbf{z}}^{\ast} = \mathbf{G}\mathbf{w}$, which simplifies the notation. Subsequently, the response vector $\mathbf{w}$ of the $L$-layer linear SIM can be optimized to minimize the \gls{LS} error between the target and output values, yielding the following standard \gls{ELM} formulation for this system design parameter:
\begin{equation}\label{eq:LS-objective}
    \mathbf{w}_{\rm LS} \triangleq \arg \min_{\mathbf{w}} \| \mathbf{z} -  \hat{\mathbf{z}} \|_{2}^{2} \stackrel{\eqref{eq:z_hat}}{=} \arg \min_{\mathbf{w}} \| \mathbf{z} -  \mathbf{G}\mathbf{w} \|_{2}^{2}.
\end{equation}
This LS problem has the following closed-form solution:
\begin{align} 
    \mathbf{w}_{\rm LS} 
    &\triangleq  \left(\mathbf{G}^{\rm H} \mathbf{G}\right)^{-1}\mathbf{G}^{\rm H} \mathbf{z} = \mathbf{G}^\dagger \mathbf{z} \in \mathbb{C}^{N_{\rm r} \times 1}, \label{eq:weights-opt-no-reg}
\end{align}
therefore, the output of the trained \gls{XL}-\gls{MIMO}-\gls{ELM} over the training set $\mathcal{D}$ may be expressed as:
\begin{align} \label{eq:trained-elm-output}
    \hat{\mathbf{z}} = \mathbf{G} \mathbf{G}^\dagger \mathbf{z},
\end{align}
Similar to standard theoretical results from \gls{ELM} and \gls{SLFN} theory, the following proposition expresses the universal approximation property of the proposed \gls{XL}-\gls{MIMO}-\gls{ELM} system framework, guaranteeing zero residual error over arbitrary data sets; therefore, $\hat{\mathbf{z}}$ becomes real-valued exactly.

\begin{proposition}[Universal Approximation]\label{prop:universal-approximation}
Consider a dataset $\mathcal{D} \triangleq \{(\mathbf{x}^{(i)}, z^{(i)})\}_{i=1}^D$ and the \gls{XL}-\gls{MIMO}-\gls{ELM} system defined by~\eqref{eq:trained-elm-output} with its preceding equations.  
In the high-\gls{SNR} regime (i.e., when $\tilde{n} \to 0$ in~\eqref{eq:output}), assuming 
$\mathbf{H}$ follows the Ricean channel model in~\eqref{eq:ricean} and the positive bias vector $\mathbf{b}$ in~\eqref{eq: C[|y|]} is drawn from a continuous distribution with positive support, there exists a weight vector 
$\mathbf{w}_{\rm LS}$ of~\eqref{eq:weights-opt-no-reg} such that 
$\|\mathbf{z} - \hat{\mathbf{z}}\|_2^2 =\|\mathbf{z} - \mathbf{G} \mathbf{w}_{\rm LS}\|_2^2 = 0$ for $N_{\rm r} = D$ with probability $1$. 
\end{proposition}

\begin{proof}
    The proof of this proposition is given in Appendix~\ref{appendix:universal-approximation-proof}.
\end{proof}

\subsubsection{Effect of the Reception Noise}

While the previous analysis and the universal approximation guarantees of Proposition~\ref{prop:universal-approximation} assume a noise-free regime to establish baseline performance, practical \gls{RX} observations are inevitably subject to the \gls{AWGN} $\tilde{n}$ in~\eqref{eq:output}.
Hence, we now quantify its impact on the predictions of the trained \gls{XL}-\gls{MIMO}-\gls{ELM}, comparing to the noise-free case.
Aggregating the scalar noise components $\tilde{n}$ across the entire dataset of size $D$, we define the noise vector $\tilde{\mathbf{n}} \sim \mathcal{CN}(\mathbf{0}, \sigma_{\tilde{n}}^2 \mathbf{I}_D)$. Consequently, the estimated target vector $\tilde{\mathbf{z}} \in \mathbb{C}^{D \times 1}$ can be modeled as the ideal target $\hat{\mathbf{z}}$ corrupted by this aggregated noise:
\begin{equation}\label{eq:noise-at-output}
    \tilde{{\mathbf{z}}} = \hat{\mathbf{z}} + \tilde{\mathbf{n}}.
\end{equation}

\begin{proposition}[Prediction error under RX noise]\label{prop:noise-effects}
Consider an ideal noise-free \gls{ELM} satisfying $\hat{\mathbf{z}} = \mathbf{G}\mathbf{w}_{\rm LS}$, and the \gls{RX} model in \eqref{eq:output} with independent \gls{AWGN} among elements of variance $\sigma_{\tilde{n}}^2$. The mean squared residual error per data point between the noise-imputed predictions $\tilde{\mathbf{z}}$ of the \gls{ELM} of~\eqref{eq:output} and the predictions $\hat{\mathbf{z}}$ of the noise-free \gls{ELM} is 
$\mathbb{E}\!\left\{(\|\tilde{\mathbf{z}}-\hat{\mathbf{z}}\|_2^2)/D\right\}=\sigma_{\tilde{n}}^2$.
\end{proposition}
\begin{proof}
    It can be easily computed that $\mathbb{E}\!\left\{\|\tilde{\mathbf{z}}-\hat{\mathbf{z}}\|_2^2\right\} \overset{\eqref{eq:noise-at-output}}{=}\mathbb{E}\!\left\{\|\tilde{
    \mathbf{n}} \|_2^2\right\}={\rm Tr}\left( \mathbb{E}\!\left\{\tilde{
    \mathbf{n}} \tilde{
    \mathbf{n}}^{\rm H}\right\}\right)=D\sigma_{\tilde{n}}^2$.
\end{proof}

Having quantified the expected system behavior in terms of the \gls{MSE} as a function of the reception noise, we consider only noise-free cases for the rest of the analytical description of this paper for ease of presentation. However, the effects of noise are quantified numerically in Section~\ref{sec:results}.
Despite Proposition~\ref{prop:noise-effects} showing a linear relationship of \gls{MSE} and $\sigma_{\tilde{n}}^2$, it is noted that \gls{MSE} is not necessarily the inference target metric of interest.
For example, for classification problems, as those considered in Section~\ref{sec:results}, the accuracy metric is less susceptible to the influence of noise, since the noise level must be large enough to change the decision threshold between the classes.

\subsection{Formulation of the Activation Matrix $\mathbf{G}$}
To obtain the \gls{LS} solution of~\eqref{eq:weights-opt-no-reg}, knowledge of $\mathbf{G}$ is needed, which contains the activated received signals for all antennas (i.e., the output of the front-end \gls{NL} \gls{MS}) and all data points.
Assuming \gls{CSI} availability at the \gls{RX}, $\mathbf{G}$ may be constructed digitally.
It is, however, possible to obtain the required observations by consecutive data measurements, bypassing the need for a dedicated \gls{CSI} estimation process~\cite{Alexandropoulos2026ReceivingRIS}.
In that regard, we propose a procedure of a low overhead, both in terms of communication resources, as well as in computational complexity, for deriving the \gls{LS} solution~\eqref{eq:weights-opt-no-reg}.
We assume that the \gls{AM} signal corresponding to each data point of $\mathcal{D}$ is transmitted for a potentially enlarged symbol duration and $M \ll N_{\rm r}$ measurements of it are collected in the following way. An arbitrary weight vector $\mathbf{a}_m \in \mathbb{C}^{N_{\rm r} \times 1}$, with $m=1,\dots,M$, is selected as the reception combiner each time.
In particular, during the $m$-th measurement time frame of the $i$-th data point ($i=1,\dots,D$), the output of the \gls{XL}-\gls{MIMO}-\gls{ELM} system at the \gls{RX} \gls{RFC} can be described as follows:
\begin{equation}\label{eq:measurement-output-single-datapoint}
    \hat{z}^{(i)}_m \triangleq \mathbf{a}_m^{\rm H} g(\mathbf{y}^{(i)}) + \tilde{n}_{i,m}, 
\end{equation}
where $\mathbf{a}_m$ plays the role of the wave-domain combiner $\mathbf{w}$ in~\eqref{eq:output}.
For convenience, let us define the measurement matrix as:
\begin{equation}
    \mathbf{A} \triangleq [\mathbf{a}^{\rm H}_1; \dots; \mathbf{a}^{\rm H}_M] \in \mathbb{C}^{M \times N_{\rm r}}.
\end{equation}

The latter process is repeated for all $D$ data points.
Owing to advancements in reconfigurable \gls{MS} design~\cite{Li2017Electromagnetic,Qian2025Nanosecond}, it is feasible to achieve \gls{MS} response configuration switching in the order of nanoseconds under suitable unit cell and control hardware.
Therefore, it is reasonable to assume that all $M$ measurements frames for each data point take place within standard symbol duration periods in the order of few microseconds, and subsequently, only $D$ transmissions are required overall. By applying the reception of~\eqref{eq:output} while replacing $\mathbf{w}$ with $\mathbf{a}_m$ and disregarding the noise term, the collected outputs for the complete data set can be compactly inserted in the following matrix:

\begin{align}\label{eq:G-sketch}
    \tilde{\mathbf{G}} \triangleq &
     \begin{bmatrix}
        g^{\rm H}(\mathbf{y}^{(1)}) \mathbf{a}_1 & \dots  & g^{\rm H}(\mathbf{y}^{(1)}) \mathbf{a}_M \\
        \vdots                                & \ddots & \vdots                                \\
        g^{\rm H}(\mathbf{y}^{(D)}) \mathbf{a}_1 & \dots  & g^{\rm H}(\mathbf{y}^{(D)}) \mathbf{a}_M
    \end{bmatrix} =
    \mathbf{G} \mathbf{A}^{\rm H} \in \mathbb{C}^{D \times M}.
\end{align}
We therefore proceed to use $\tilde{\mathbf{G}}$ instead of $\mathbf{G}$ in~\eqref{eq:LS-objective}, 
obtaining, following \eqref{eq:weights-opt-no-reg}'s \gls{LS} solution, the intermediate weight vector $\boldsymbol{\alpha}_{\rm sketch} \in \mathbb{C}^{M \times 1}$ in the reduced $M$-dimensional subspace:
\begin{align}
    \boldsymbol{\alpha}_{\rm sketch} &=
    \tilde{\mathbf{G}}^\dagger \mathbf{z} \in \mathbb{C}^{M \times 1}. \label{eq:alpha-opt-sketched-no-reg}
\end{align}
Then, to synthesize the physical $N_{\rm r}$-dimensional reception combining weight vector required by our system for single-shot inference, we project $\boldsymbol{\alpha}_{\rm sketch}$ back via the measurement codebook:
\begin{equation}
    \mathbf{w}_{\rm sketch} \triangleq \mathbf{A}^{\rm H} \boldsymbol{\alpha}_{\rm sketch} = \mathbf{A}^{\rm H} \tilde{\mathbf{G}}^\dagger \mathbf{z} \in \mathbb{C}^{N_{\rm r} \times 1}. \label{eq:weights-opt-sketched}
\end{equation}
We use the term ``sketching operator'' for $\mathbf{A}$, borrowing the framework from randomized linear algebra, and therefore, $\tilde{\mathbf{G}}$ and $\mathbf{w}_{\rm sketch}$ will be referred to as the sketched observation matrix and weight vector, respectively.
Subsequently, the output of the proposed \gls{XL}-\gls{MIMO}-\gls{ELM} system once the sketching weights are determined can be expressed, similar to \eqref{eq:trained-elm-output}, as:
\begin{equation}
\hat{\mathbf{z}}_{\rm sketch} 
\triangleq
\mathbf{G} \mathbf{w}_{\rm sketch}
\overset{\eqref{eq:weights-opt-sketched}}{=}
\mathbf{G}  \mathbf{A}^{\rm H} \tilde{\mathbf{G}}^\dagger \mathbf{z}
\overset{\eqref{eq:G-sketch}}{=} 
\tilde{\mathbf{G}} \tilde{\mathbf{G}}^\dagger \mathbf{z}.\label{eq:sketch-output}
\end{equation}

The motivation for the given form for $\mathbf{w}_{\rm sketch}$ in~\eqref{eq:weights-opt-sketched} is that, when the full measurement dimension $M=N_{\rm r}$ is used, the resulting combining vector coincides with the original \gls{LS} solution from~\eqref{eq:weights-opt-no-reg}, and therefore, the sketched version of the \gls{XL}-\gls{MIMO}-\gls{ELM} leads to the same predictions as if the sketching operation were not in place. This is formalized below.

\begin{proposition}[Exact \gls{LS} weight recovery for $M=N_r$]\label{prop:equal-weights-for-Nr}
Let $M=N_{\rm r}$ and let $\mathbf{A}\in\mathbb{C}^{N_{\rm r}\times N_{\rm r}}$ be full rank. Then, the sketching weight vector in~\eqref{eq:weights-opt-sketched} satisfies $\mathbf{w}_{\mathrm{sketch}} = \mathbf{w}_{\mathrm{LS}}$.
\end{proposition}

\begin{proof}
Since $\mathbf{G}$ has full column rank (under the assumption of rich scattering conditions, discussed in Proposition~\ref{prop:universal-approximation}) and $\mathbf{A}$ is invertible, the pseudoinverse expression may be simplified as:
\begin{equation}
\widetilde{\mathbf{G}}^\dagger
=
\left(\mathbf{G}\mathbf{A}^{\rm H}\right)^\dagger
=
\left(\mathbf{A}^{\rm H}\right)^\dagger
\mathbf{G}^\dagger
=
\mathbf{A}\left(\mathbf{A}^{\rm H}\mathbf{A}\right)^{-1}
\mathbf{G}^\dagger.
\label{eq:sketch_ls_reverse_order}
\end{equation}
Substituting~\eqref{eq:sketch_ls_reverse_order} into~\eqref{eq:weights-opt-sketched} yields:
\begin{align}\label{eq:sketch_ls_weight_derivation}
\mathbf{w}_{\mathrm{sketch}}
=
\mathbf{A}^{\rm H}\mathbf{A}
\left(\mathbf{A}^{\rm H}\mathbf{A}\right)^{-1}
\mathbf{G}^\dagger\mathbf{z} =
\mathbf{G}^\dagger\mathbf{z}
\overset{\eqref{eq:weights-opt-no-reg}}{=}
\mathbf{w}_{\mathrm{LS}},
\end{align}
which completes the proof.
\end{proof}

\textit{Corollary.} Since $\mathbf{w}_{\rm LS} = \mathbf{w}_{\rm sketch}$ for $M = N_{\rm r}$, it follows from~\eqref{eq:trained-elm-output} and~\eqref{eq:sketch-output} that the output of the sketched and the original versions of the proposed \gls{XL}-\gls{MIMO}-\gls{ELM} system are exactly the same, i.e., $\hat{z} = \hat{z}_{\rm sketch}$.

\begin{proposition}[Monotonicity of prediction error]\label{prop:monotonicity}
For any nested sketch sizes $M_1 \le M_2$, with corresponding sketch outputs $\hat{\mathbf{z}}_{\rm sketch}(M_1)$ and $\hat{\mathbf{z}}_{\rm sketch}(M_2)$ following~\eqref{eq:sketch-output}, the sketching residual prediction error is monotonically non-increasing, i.e.:
\begin{equation}
    \left\| \hat{\mathbf{z}} - \hat{\mathbf{z}}_{\rm sketch}(M_2) \right\|_2 \le \left\| \hat{\mathbf{z}} - \hat{\mathbf{z}}_{\rm sketch}(M_1) \right\|_2
\end{equation}
\end{proposition}

\begin{proof}
    This proposition's proof is delegated in Appendix~\ref{appendix:monotonicity-proof}.
\end{proof}

Propositions~\ref{prop:equal-weights-for-Nr} and~\ref{prop:monotonicity} demonstrate that the sketching procedure leads to approximately equal performance with sufficiently large measurement dimension $M$ (and fully equal performance for $M = N_{\rm r}$), and establish the fundamental design criteria for the sketching matrix $\mathbf{A}$. 
According to Proposition~\ref{prop:equal-weights-for-Nr}, $\mathbf{A}$ must be full rank. 
Consequently, randomized matrices or subsets of the \gls{DFT} matrix (that are typical combiner choices~\cite{Yuan2023DFT}) satisfy this condition. 
Furthermore, Proposition~\ref{prop:monotonicity} dictates that nested sketching structures guarantee a non-increasing residual error. Therefore, progressively expanding \gls{DFT} matrices or nested randomized unit-amplitude phase-shifting matrices represent structurally sound and optimal choices for the measurement collection matrix $\mathbf{A}$.

\subsection{Wave-Domain Combining with Linear SIM}

In the previous sections, we provided insights on the desired values of the reception combiner $\mathbf{w}$, the sketching combining weights $\mathbf{A}$, as well as the more practical choice of $\mathbf{w}_{\rm sketch}$. Whether these weights can be realized by the response of the RX's $L$-layer linear \gls{SIM} remains a yet open problem. In fact, several approaches in the literature have focused on optimizing the responses of linear \glspl{SIM} to approximate certain functionalities, 
such as the \gls{DFT}~\cite{AYG24_SIM_DFT} and arbitrary hidden neural network layers~\cite{Gunduz_Layer_Approximation, Pandolfo_SIM_semantic_alignment}.
However, all those approaches rely on iterative gradient-based optimization algorithms and, most importantly, do not guarantee that the desired responses will be achieved with sufficiently small error. In this section, we present an analytical approach for the configuration of the linear \gls{SIM} 
that guarantees exact solution to the desired combining weights $\mathbf{w}$, thus, realizing entirely \gls{OTA} the inference output defined in~\eqref{eq:sketch-output}. 

Let us assume that the responses $\{\boldsymbol{\phi}_l\}_{l=1}^{L-1}$ of the $L-1$ linear \gls{MS} layers at the RX, following the front-end \gls{NL} \gls{MS} layer indexed with $l=0$ and excluding the final $L$-indexed linear \gls{MS} layer, are randomly chosen.
The diffracted signal $g(\mathbf{y})$ from the \gls{NL} layer undergoes a series of linear transformations affected by the deterministic propagation model described in~\eqref{eq:sim-propagation}. For convenience, let $\mathbf{\Omega} \triangleq \mathbf{\Psi}_L \prod_{l=L-1}^{1} \mathbf{\Phi}_l \mathbf{\Psi}_l   \in \mathbb{C}^{N_{\rm r} \times N_{\rm r}}$ denote the response of the system including the propagation between the $l=0$ (\gls{NL}) and $l=1$ (linear) \gls{MS} layers up to the propagation between the $L-1$ and $L$ linear layers.
We henceforth focus on tuning exclusively the amplitude and phase responses of the final $L$-indexed linear \gls{MS} layer. In particular, the objective is to optimize $\mathbf{\Phi}_L$ such that the response of the entire $L$-layer linear \gls{SIM}
emulates the reception combining vectors $\mathbf{a}_m$'s and  $\mathbf{w}_{\rm sketch}$ (or even $\mathbf{w}_{\rm LS}$) that are used by the proposed \gls{XL}-\gls{MIMO}-\gls{ELM} framework. As shown below, this can be achieved analytically.

\begin{proposition}[Closed-form linear \gls{SIM} configuration for wave-domain combining]\label{prop:closed-form-MS}
    Let $\mathbf{\Omega} \in \mathbb{C}^{N_{\rm r} \times N_{\rm r}}$ be the full-rank propagation matrix among the \gls{MS} layers indexed with $l=0,\ldots,L-1$, 
    and let $[\boldsymbol{\omega}]_j \neq 0$ $\forall j=1,\dots,N_{\rm r}$ for the wave propagation vector in~\eqref{eq:rayleigh-somerfeld-vector} from the final $L$-indexed linear \gls{MS} layer to the single antenna element attached to the RX \gls{RFC}. For any arbitrary combining vector $\mathbf{w} \in \mathbb{C}^{N_{\rm r} \times 1}$, the phase and amplitude shifts $\boldsymbol{\phi}_L^{\mathbf{w}}$ of the final \gls{MS} layer that exactly match the noise-free physical received signal to the scaled target functional $\gamma \mathbf{w}^{\rm H} g(\mathbf{y})$ are given by the following closed-form configuration:
\begin{equation}\label{eq:sim-closed-form}
    \boldsymbol{\phi}_L^{\mathbf{w}} = \gamma \left[ \frac{\left[ \mathbf{w}^{\rm H} \mathbf{\Omega}^{-1} \right]_1}{[\boldsymbol{\omega}]_1}, \dots, \frac{\left[ \mathbf{w}^{\rm H} \mathbf{\Omega}^{-1} \right]_{N_{\rm r}}}{[\boldsymbol{\omega}]_{N_{\rm r}}} \right]^\top,
\end{equation}
\textit{where $\gamma \in \mathbb{R}^+$ is a global scaling factor satisfying $\gamma \le \frac{1}{\zeta}$ with $\zeta \triangleq \max_n \left( \left| \left[ \mathbf{w}^{\rm H} \mathbf{\Omega}^{-1} \right]_n/{[\boldsymbol{\omega}]_n} \right| \right)$.}
\end{proposition}
\begin{proof}
The proof of this proposition is given in Appendix~\ref{appendix:closed-form-MS-proof}.  
\end{proof}

Proposition~\ref{prop:closed-form-MS} makes the assumption that $\mathbf{\Omega}$ is full rank.
This is ascertained via the following result.
\begin{proposition}[Full-rankness of the linear \gls{SIM} matrix $\boldsymbol{\Omega}$]\label{prop:sim-full-rank}
    Let $\{\boldsymbol{\phi}_l\}_{l=1}^{L-1}$ be arbitrarily sampled from an absolutely continuous distribution.
    Assuming all \glspl{MS} elements spatial coordinates are configured such that they do not belong to the measure-zero set of degenerate placements, then ${\rm rank}(\mathbf{\Omega}) = N_{\rm r}$ almost surely.
\end{proposition}
\begin{proof}
The proposition's proof is presented in Appendix~\ref{appendix:sim-full-rank-proof}.    
\end{proof}

\begin{remark}
Proposition~\ref{prop:sim-full-rank} implies that the minimum number of \gls{MS} layers for $\mathbf{\Omega}$ to be full rank is $1$. Hence, for the considered inference operations, $L=2$ suffices, i.e., \gls{RX}'s heterogeneous \gls{SIM} at the \gls{XL}-\gls{MIMO}-\gls{ELM} system needs only the front-end \gls{NL} \gls{MS}, an intermediate linear \gls{MS} with fixed random responses, and a third linear \gls{MS} layer with controllable responses using~\eqref{eq:sim-closed-form}.
\end{remark}

\begin{algorithm}[!t]
\caption{Training with the Proposed XL-MIMO-ELM}
\label{alg:training}
\begin{algorithmic}[1]
\Statex \textbf{Inputs:} Dataset $\mathcal{D} = \{(\mathbf{x}^{(i)}, z^{(i)})\}_{i=1}^D$, propagation matrix $\boldsymbol{\Omega}$, sketching codebook $\mathbf{A}$, and arbitrary \gls{AM} phase shifts $\boldsymbol{\psi}$.
\For{$i=1,\dots,D$}
\State The TX constructs signal $\bar{\mathbf{x}}^{(i)}$ via \eqref{eq:transmit-signal} and transmits it.
\For{$m=1,\dots,M$}
\State The RX selects a combining vector $\mathbf{a}_m = ([\mathbf{A}]_{m,:})^{\rm H}$,
\State configures the $L$-indexed linear \gls{MS} as $\boldsymbol{\phi}^{\mathbf{a}_m}_L$ via~\eqref{eq:sim-closed-form},
\State observes $\hat{z}^{(i)}_m\!=\! \boldsymbol{\omega}{\rm diag}(\boldsymbol{\phi}^{\mathbf{a}_m}_L)\mathbf{\Omega}g(\mathbf{y}^{(i)})\!+\!\tilde{n}^{(i)}_m$, 
and\label{alg:training:line:observe} 
\State computes $[\tilde{\mathbf{G}}]_{i,m} = (\hat{z}^{(i)}_m)^\ast / \gamma $ to construct ~\eqref{eq:G-sketch}.
\EndFor
\EndFor
\State The RX calculates the weights $\mathbf{w}_{\rm sketch}$ via~\eqref{eq:weights-opt-sketched}, and computes the \gls{MS} configuration $\boldsymbol{\phi}^{\mathbf{w}_{\rm sketch}}_L$ via~\eqref{eq:sim-closed-form}. 
\Statex \textbf{Output:} 
Trained response $\boldsymbol{\phi}^{\mathbf{w}_{\rm sketch}}_L$ for the $L$-indexed linear MS layer to be used for inference.
\end{algorithmic}
\end{algorithm}

\section{XL-MIMO-ELM Training and Deployment}\label{sec:training-deployment}

Having described all particularities of the proposed \gls{XL}-\gls{MIMO}-\gls{ELM} system architecture, we herein provide a complete description of how it implements \gls{EI}, in particular, its training and deployment/inference phases. 

We commence with the training process summarized in Algorithm~\ref{alg:training}. 
As observed, this process involves two main stages.
During the first stage, all data points are transmitted once and repeatedly observed $M$ times at the RX, with changing configurations at the final $L$-indexed linear MS layer, to collect the sketching measurements. In the second stage, the optimal combining weights are first calculated, and then the response of the $L$-th layer is finally configured so that the overall $L$-layer linear \gls{SIM} approximates the \gls{LS} weight vector in~\eqref{eq:weights-opt-sketched}. This training process requires $O(D)$ transmissions,  
necessitating $M+1$ reconfigurations of the $L$-th linear MS layer. Note that no explicit knowledge of $\mathbf{H}$ is required and a single RX \gls{RFC} is used for signal observation.
Since the configurations of the first $L-1$ linear \gls{MS} layers may be selected to be static, and they are independent of the training data and \gls{CSI}, we assume that the propagation quantities $\boldsymbol{\Omega}$ and $\boldsymbol{\omega}$ in the RX's front-end have been accurately measured during the structure's manufacturing process. 
Computationally, the complexity of Algorithm~\ref{alg:training} is dominated by the \gls{LS} solution, which is $\Theta(D M^2)$, assuming it is solved through matrix factorization.
Since $M \ll D$, the computational overhead is therefore quasi-linear to the dataset size.
Considering the low training overhead of the overall methodology, it is reasonable to assume that training may take place within the channel coherence time.
As a simple feasibility example, a bandwidth of $B=180$~KHz corresponds to a symbol duration of approximately $T_s = 5.6$~$\mu$s.
Hence, a typical coherence interval of $T_c=5$~ms accommodates approximately $T_c/T_s\cong893$ transmissions.
Since the $N_{\rm t}$ features of each data point are transmitted in parallel, this is sufficient for training with reasonably small data sets, if the $M$ measurements via wave-domain combining are collected within each symbol interval.


\begin{algorithm}[!t]
\caption{Inference with the Proposed XL-MIMO-ELM}
\label{alg:inference}
\begin{algorithmic}[1]
\Statex \textbf{Inputs:} Arbitrary \gls{AM} phase shifts $\boldsymbol{\psi}$, propagation matrix $\boldsymbol{\Omega}$, configuration of the $L$-indexed linear MS layer $\boldsymbol{\phi}^{\mathbf{w}_{\rm sketch}}_L$ from Algorithm~\ref{alg:training}, and input data point $\mathbf{x}$ for inference.
\State The TX constructs signal $\bar{\mathbf{x}}$ via \eqref{eq:transmit-signal} and transmits it, while the RX uses the trained configuration $\boldsymbol{\phi}^{\mathbf{w}}_L$ to implement $\mathbf{w}_{\rm sketch}$ entirely in the wave domain via its $L$-layer linear \gls{SIM}.
\State The RX observes $\hat{z} = \mathbf{w}_{\rm sketch}^{\rm H} g(\mathbf{y}) + \tilde{n}$. 
\label{alg:inf:line:output}
\Statex \textbf{Output:} Inference result $\hat{z}$.
\end{algorithmic}
\end{algorithm}  

The inference process with the proposed \gls{XL}-\gls{MIMO}-\gls{ELM} system has an even lower overhead.
For a single inference data point $\mathbf{x}$, this process is detailed in Algorithm~\ref{alg:inference}. Note that only a single transmission is now required and no sketching matrices or digital computations are required, since the optimal combiner has been computed and implemented during training in the wave domain through the linear \gls{SIM}. This implies that the inference complexity is of pseudo-constant time. 

It is finally noted that, in the proposed \gls{XL}-\gls{MIMO}-\gls{ELM} framework, both training and inference take place within the same channel coherence block.
In practice, it is reasonable for the system to apply inference over batches of new data to avoid frequent re-training steps under different coherence blocks.
Besides, with well regularized models, our framework is expected to be able to overcome small fluctuations of the XL MIMO channel $\mathbf{H}$ and also to be able to be rapidly re-tuned to time-evolving channels, e.g., similar to the methodology of~\cite{Stylianopoulos_MIMO_ELM}.
We leave such investigations to future work.

\section{Numerical Results and Discussion}\label{sec:results}

\begin{figure*}[t]
    \centering
    \begin{subfigure}[b]{0.32\textwidth}
        \centering
        \includegraphics{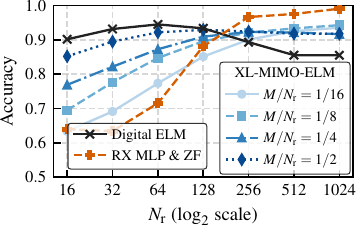}
        \caption{WBCD.}
        \label{fig:results_breast_cancer}
    \end{subfigure}
    \hfill
    \begin{subfigure}[b]{0.32\textwidth}
        \centering
        \includegraphics{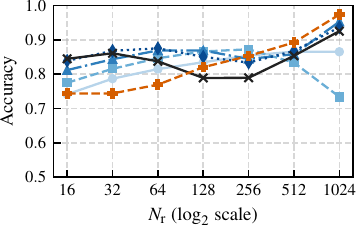}
        \caption{Parkinson's.}
        \label{fig:results_parkinsons}
    \end{subfigure}
    \hfill
    \begin{subfigure}[b]{0.32\textwidth}
        \centering
        \includegraphics{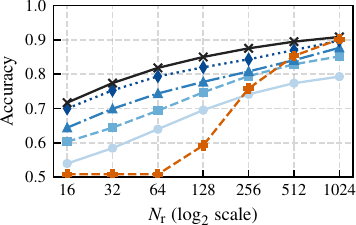}
        \caption{MNIST.}
        \label{fig:results_mnist}
    \end{subfigure}
    \caption{Classification accuracy the proposed \gls{XL}-\gls{MIMO}-\gls{ELM} approach versus the number of metamaterials $N_{\rm r}$ (corresponding to the number of trainable parameters) at the \gls{RX}'s \gls{NL} \gls{MS} layer, considering three distinct datasets and different sketching values $M$.}
    \label{fig:overall_results}
\end{figure*}

\begin{figure}[t]
    \centering
    \begin{subfigure}{\columnwidth}
        \centering
        \includegraphics[width=0.8\linewidth]{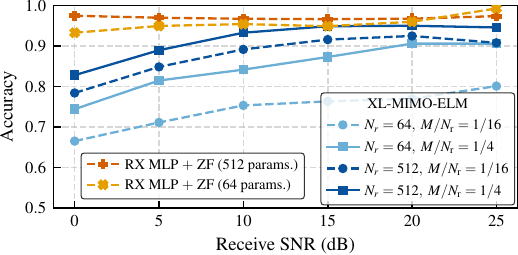}
        \caption{WBCD}
    \end{subfigure}
    \vspace{0.05cm}
    
    \begin{subfigure}{\columnwidth}
        \centering
        \includegraphics[width=0.8\linewidth]{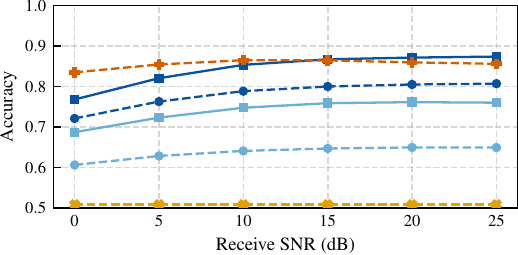}
        \caption{MNIST}
    \end{subfigure}

    \caption{Performance of the \gls{XL}-\gls{MIMO}-\gls{ELM} approach and digital \gls{EI} benchmark versus different levels of the measurement \gls{SNR}.}
    \label{fig:results-snr}
\end{figure}

\begin{figure}[t]
    \centering
    \begin{subfigure}{\columnwidth}
        \centering
        \includegraphics[width=0.8\linewidth]{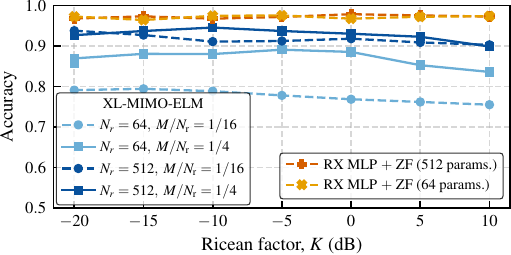}
        \caption{WBCD}
    \end{subfigure}
    \vspace{0.05cm}

    \begin{subfigure}{\columnwidth}
        \centering
        \includegraphics[width=0.8\linewidth]{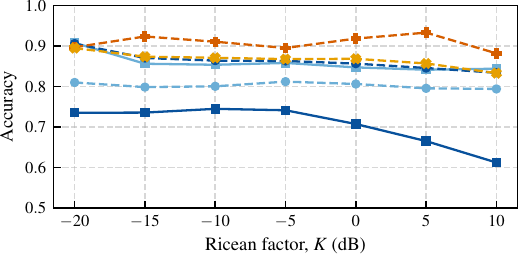}
        \caption{Parkinson's}
    \end{subfigure}

    \caption{Classification accuracy of the \gls{XL}-\gls{MIMO}-\gls{ELM} approach and digital \gls{EI} benchmark over different Ricean factors.}
    \label{fig:results-ricean}
\end{figure}

We have evaluated the proposed \gls{XL}-\gls{MIMO}-\gls{ELM} framework for 
\gls{EI} on several standard small-to-medium-size binary classification datasets under a range of system settings and channel conditions. The considered datasets are the Parkinson's~\cite{Parkinsons_dataset} and the \gls{WBCD}~\cite{breast_cancer_dataset} of $195$ and $569$ data points of $22$ and $30$ numerical features (corresponding to $N_{\rm t}$ values) for disease diagnosis, respectively, as well as the MNIST dataset of $6\times 10^4$ handwritten digit images~\cite{MNIST} (subsampled to $100$ pixels  
and converted the objective to even/odd digit classification).
An $80\!:\!20$ train-test split was imposed, and all input features were normalized to the $(0,1)$ range.
To determine the predicted class, we simply applied the discretization ${\rm sign}(\mathfrak{Re}(\hat{z}))$ of~\eqref{eq:output}, and we report in the performance figures that follow the mean accuracy values over $200$ random initializations. 
For convenience, the phases ${\bm \psi}$ of the \gls{AM} signals were set identically to $\mathbf{0}$, and $\mathbf{b}$ was sampled from a Rayleigh distribution with scale parameter $\mathbb{E}[\|\mathbf{H}\|_{\rm F}]/(2 N_{\rm r} N_{\rm t})$, for it to be in the same order as $|\mathbf{y}|$.
Unless otherwise specified, the channel parameters were set as $P_L = -50$ dB and $K=0$ dB.
For the linear \gls{SIM}, we employed $L=2$ layers, each consisting of $ \sqrt{N_{\rm r}} \times \sqrt{N_{\rm r}}$ diffractive elements.
To overcome overfitting, we have used the regularized versions of the pseudo-inverse matrices, i.e., we compute $(\tilde{\mathbf{G}}^{\rm H}\tilde{\mathbf{G}} +\ell \mathbf{I})^{-1}\tilde{\mathbf{G}}^{\rm H}$ instead of $\tilde{\mathbf{G}}^\dagger$ from~\eqref{eq:weights-opt-sketched} onward with $\ell=10^{-6}$.

We have considered two inference baselines. The first is a theoretical \gls{ELM} model 
with an equal number of hidden layer size as the proposed \gls{XL}-\gls{MIMO}-\gls{ELM}.
This approach performs traditional classification on the (real-domain) datasets, disregarding the system model and its effects, therefore, it is henceforth termed as ``Digital ELM.'' The hidden layer parameters were sampled from the standard normal distribution.
The second approach is an \gls{EI} approach where the \gls{TX}-side processing remains as described, and assumed idealistically full  knowledge of \gls{CSI} and the signal $\mathbf{y}$ at each of the elements of the front-end \gls{NL} \gls{MS} in the RX. Therefore, $\bar{\mathbf{x}}$ was estimated via \gls{ZF} as $\bar{\mathbf{x}} \cong \mathbf{H}^\dagger\mathbf{y}$ and a $2$-layered \gls{MLP} classifier was trained at the \gls{RX} across random channel realizations. The total number of parameters of the \gls{MLP} has been set approximately equal to that of the \gls{ELM} for each case. We note that this baseline is associated with increased system knowledge requirements and excessive training times.

The performance of the proposed training approaches for an increasing number of elements $N_{\rm r}$ at the \gls{RX}'s front-end \gls{NL} \gls{MS} layer is displayed in Fig.~\ref{fig:overall_results}.
Note that $N_{\rm r}$ also corresponds to the number of trainable \gls{ELM} 
parameters, which directly impacts the models' approximation capacity.
Different sketching matrix sizes $M$ have been used, so that $M/N_{\rm r}$ remains constant in $\{1/16, 1/8, 1/4, 1/2\}$ for each case. 
Firstly, it can be observed that, as $N_{\rm r}$ increases, the classification accuracy generally improves for all $M$ values, excluding some limited overfitting cases due to over-parameterization.
Despite ``Digital ELM'' being an ideal theoretical model unaffected by channel and system conditions and ``RX MLP'' applying \gls{ZF} to negate channel transformations, the proposed \gls{XL}-\gls{MIMO}-\gls{ELM} generally approaches the performance or even outperforms the digital baselines, given the same numbers of parameters.
This aligns with the result of Proposition~\ref{prop:universal-approximation} (assuming $M=N_{\rm r})$, where for $N_{\rm r} \to D$, the \gls{XL}-\gls{MIMO}-\gls{ELM} achieves zero error in dataset predictions.
It is noted that Proposition~\ref{prop:universal-approximation} considers training set errors in terms of \gls{MSE}, whereas the reported results consider classification accuracy on the test set.
Nevertheless, training \gls{MSE} is empirically a good proxy for the desired test set accuracy, which measures generalization.
As a secondary point, Fig.~\ref{fig:overall_results} allows for investigation of the effects of the dimension $M$ of the sketching matrix. 
Higher $M$ values generally lead to better results, as indicated by Proposition~\ref{prop:monotonicity}.
Again, note that Proposition~\ref{prop:monotonicity} considers specifically the training \gls{MSE} measured over the same channel realizations; subsequently, small variations, such as those of Fig.~\ref{fig:results_parkinsons}, are expected due to overfitting effects, without changing the overall trend of the proposed system's behavior.
Finally, it is important to highlight that, as $N_{\rm r}$ increases, the differences in the achieved performance across different $M$ values are diminishing, since the \gls{XL}-\gls{MIMO}-\gls{ELM} becomes over-parametrized, and is therefore less susceptible to the partial observability of the features of the hidden layer that are obtained through the sketching/repeated combining process.

In the previous investigation, we have assumed perfect knowledge of $\tilde{\mathbf{G}}$ during training, i.e., we have manually set $\tilde{n}=0$ in~\eqref{eq:output} (and equivalently in lines~\ref{alg:training:line:observe} and~\ref{alg:inf:line:output} of Algorithms~\ref{alg:training} and~\ref{alg:inference}, respectively), to test empirical agreement of the numerical results with the theoretical results of Section~\ref{sec:noise_free}, which were derived under a noise-free assumption. 
Subsequently, we proceed to investigate the effects of reception noise.
Figure~\ref{fig:results-snr} reports the performance across different receive \gls{SNR} levels, computed as $\mathbb{E}[|\mathbf{w}^{\rm H}g(\mathbf{y})|^2/\sigma_{\tilde{n}}^2]$. 
Note that the \gls{AWGN} of~\eqref{eq:output} is applied both during reception of the training set, i.e., when constructing $\tilde{\mathbf{G}}$ through the repeated wave-domain combining~\eqref{eq:G-sketch}, as well as during inference on the test set, i.e., when the aforementioned algorithmic lines are executed by the system.
It is shown that, while performance degrades at low measurement \gls{SNR}, the \gls{XL}-\gls{MIMO}-\gls{ELM} converges to the ideal noise-free \gls{SNR} results, validating the idealized assumptions of our theoretical analysis.
Evidently, larger $N_{\rm r}$ and $M$ values lead to better performance. 
The ``RX MLP'' baseline is understandable much less affected by noise, since the channel equalization is handled by \gls{ZF}.

So far, we have assumed rich scattering conditions (i.e., $K=0$) for the \gls{XL}-\gls{MIMO} channel $\mathbf{H}$, ensuring that it is a full-rank matrix. This improves the universal approximation capability of our \gls{XL}-\gls{MIMO}-\gls{MIMO} framework since the condition number of $\mathbf{G}$ decreases, yielding better conditioning for the \gls{LS} solution. Figure~\ref{fig:results-ricean} illustrates the impact of channel scattering considering various Ricean factors $K$.
As depicted, relatively stable performance is maintained across the examined values of $M$ and $N_{\rm r}$ under sufficiently diverse channels.
However, as the channel becomes \gls{LoS}-dominant (i.e., $K$ increases), classification accuracy degrades because the columns of $\mathbf{G}$ become linearly dependent, and therefore, the conditions for universal approximation are no longer satisfied. As observed in the figure, this is more pronounced in the Parkinson's data set with $N_{\rm r}=512$, where the overparametrized model is more susceptible to overfitting.


\section{Conclusions and Future Work}\label{sec:conclusion}
This paper showcased that \gls{XL} \gls{MIMO} systems with properly designed \gls{MS} components at the \gls{RX} side can perform computations equivalent to a complex-valued \gls{ELM}, enabling \gls{ML}-based inference on the transmitted data entirely \gls{OTA}. The channel coefficients of the proposed \gls{XL}-\gls{MIMO}-\gls{ELM} system serve as random hidden-layer weights, while its \gls{RX}'s heterogeneous NL/linear \gls{SIM} implements the activation function, bias, and trainable weights. 
The \gls{ELM} mathematical treatment allows for closed-form training, that may be even implemented within the channel coherence time for reasonably sized data sets, while the theoretical analysis confirmed our system's universal approximation capability.
To tackle practical challenges during deployment, two key additions have been made to the standard \gls{ELM} framework.
The first one considers a codebook of RX combining configurations associated with a data transmission protocol to reduce the number of transmissions required for observing the training data with a single antenna/\gls{RFC}, which led to a modification of the \gls{ELM} optimization procedure.
The second presents a novel 
analytical approach for configuring the responses of the last linear \gls{MS} layer 
so that any desired combining vector is implemented directly in the wave domain. Thus, the overall \gls{XL}-\gls{MIMO}-\gls{ELM} system architecture operates exclusively in analog, albeit it requires digital processing for computing the
\gls{LS} solution and the configuration weights of a linear \gls{MS} during training. To further empower the proposed framework, future work will focus on iterative training procedures that optimize the \gls{ELM} across time-evolving environments~\cite{Stylianopoulos_MIMO_ELM}, thus reducing the amortized training overhead.
Additionally, the framework is suitable for integration to wider \gls{MS}-based \gls{EI} frameworks~\cite{Stylianopoulos_GO}, that achieve higher expressivity than \glspl{SLFN}.
In showcasing the applicability for this framework and its extensions, distributed sensing applications provide an exciting direction.

\begin{appendices}

\section{Used Definitions, Propositions, and Lemmas}\label{appendix:definitions}

\begin{definition}[Real Analytic Function~{\normalfont
\cite[Section~2.2]{KrantzParks2002}}]\label{def:real-analytic-function}
A function $f:\mathcal{I}\to\mathbb{R}$ defined on an open set $\mathcal{I}\subseteq\mathbb{R}^m$ is real analytic if, for every point $\mathbf{x}_0\in\mathcal{I}$, there exists a neighborhood $\mathcal{U}$ of $\mathbf{x}_0$ in which $f$ can be represented by a multivariate power series that converges to $f(\mathbf{x})$ $\forall\mathbf{x}\in\mathcal{U}$.
\end{definition}

\begin{definition}[Lebesgue Measure~{\normalfont
\cite[Section~2]{KrantzParks2002}}]\label{def:lebesgue-measure}
The Lebesgue measure on $\mathbb{R}^m$, denoted by $\lambda^m$, is the standard translation-invariant measure assigning $m$-dimensional volume to measurable sets in $\mathbb{R}^m$.
\end{definition}

\begin{definition}[Zero Set~{\normalfont
\cite[Proposition~0]{Mityagin2020}}]\label{def:zero-set}
For a function $f: \mathcal{I} \to \mathbb{R}$, its zero set is defined as $\mathcal{Z}(f) \triangleq \{\theta \in \mathcal{I} : f(\theta) = 0\}$.
\end{definition}

\begin{proposition}[Identity Theorem for Real Analytic Functions~{\normalfont \cite[Proposition~3]{Mityagin2020}}]\label{prop:analytic-zero-set}
Let $\mathcal{I} \subset \mathbb{R}^m$ be an open connected domain and let $f: \mathcal{I} \to \mathbb{R}$ be a real analytic function. If $f \not\equiv 0$ on $\mathcal{I}$ (i.e., $\exists \theta_0 \in \mathcal{I}$ such that $f(\theta_0) \neq 0$), then its zero set $\mathcal{Z}(f)$ has Lebesgue measure zero, i.e., $\lambda^m(\mathcal{Z}(f)) = 0$.
\end{proposition}

\begin{lemma}[Analyticity under Composition~{\normalfont
\cite[Section~1.4]{KrantzParks2002}}]\label{lem:analytic-composition}
Let $\mathcal{U} \subset \mathbb{R}^m$ and $\mathcal{V} \subset \mathbb{R}^p$ be open sets. If $g: \mathcal{U} \to \mathcal{V}$ and $f: \mathcal{V} \to \mathbb{R}$ are real analytic functions, then the composite function $(f \circ g): \mathcal{U} \to \mathbb{R}$ is real analytic on $\mathcal{U}$.
\end{lemma}

\begin{proposition}[Analyticity of Elementary Functions~{\normalfont
\cite[Sections~1,~2]{KrantzParks2002}}]\label{prop:elementary-analytic}
Polynomials, exponential functions, affine transformations, and the softplus activation function $\sigma(x) = \log(1 + e^x)$ are real analytic on their respective open domains.
\end{proposition}


\begin{definition}[Open Connected Domain~{\normalfont
\cite[Sections~23,~25]{Munkres2000}}]\label{def:open-connected-domain}
A subset $\mathcal{I} \subseteq \mathbb{R}^m$ is an open connected domain if it is open (every point has an open neighborhood contained in $\mathcal{I}$) and path-connected (any two points in $\mathcal{I}$ can be joined by a continuous path inside $\mathcal{I}$).
\end{definition}

\begin{definition}[Real Codimension~{\normalfont
\cite[Sections~2C,~3E]{AxlerLADR2024}}]\label{def:real-codimension}
The real codimension of a linear or affine subspace $\mathcal{V} \subset \mathbb{R}^m$ is defined as $\operatorname{codim}_{\mathbb{R}}(\mathcal{V}) \triangleq m - \dim_{\mathbb{R}}(\mathcal{V})$, with $\dim_{\mathbb{R}}(\cdot)$ denoting the dimension of a real space.
\end{definition}

\begin{definition}[Isomorphism between $\mathbb{C}^k$ and $\mathbb{R}^{2k}$~{\normalfont
\cite[Sections~2C,~6]{AxlerLADR2024}}]\label{def:complex-real-isomorphism}
The canonical real vector space isomorphism $\psi: \mathbb{C}^k \to \mathbb{R}^{2k}$ maps $\mathbf{z} = \mathbf{x} + i\mathbf{y} \in \mathbb{C}^k$ to $(\mathbf{x}, \mathbf{y}) \in \mathbb{R}^{2k}$, preserving topological openness, connectedness, and real analyticity.
\end{definition}
\begin{proposition}[Preservation of Connectedness under Codimension $\geq 2$ Removal~{\normalfont\cite[Theorem~1.1]{ConradHyperplaneComplements}}
]
\label{prop:codim-connectedness}
Let $\mathcal{U} \subset \mathbb{R}^m$ be an open connected set, and let
$\mathcal{V} \subset \mathbb{R}^m$ be a linear subspace with
$\operatorname{codim}_{\mathbb{R}}(\mathcal{V}) \geq 2$.
Then, the complement set $\mathcal{U} \setminus \mathcal{V}$ remains open and
connected in $\mathbb{R}^m$.
\end{proposition}

\begin{proposition}[Convexity Implies Path Connectedness~{\normalfont\cite[Section~25]{Munkres2000}}]
\label{prop:convex-connected}
Any nonempty convex subset $\mathcal{C} \subseteq \mathbb{R}^m$
(such as the positive orthant
$\{\mathbf{b} \in \mathbb{R}^{N_{\rm r}} :
[\mathbf{b}]_j > 0\, \forall j\}$)
is path-connected.
\end{proposition}

\begin{proposition}[Real Codimension of Complex Linear Subspace~{\normalfont\cite[Sections~1B and~2C]{AxlerLADR2024}}
]
\label{prop:complex-codim}
Let $\mathbf{a} \in \mathbb{C}^k \setminus \{\mathbf{0}\}$.
The complex linear constraint
$\mathbf{a}^{\rm H}\mathbf{z}=0$ over $\mathbb{C}^k$
induces two independent real linear constraints
$\operatorname{Re}(\mathbf{a}^{\rm H}\mathbf{z})=0$ and
$\operatorname{Im}(\mathbf{a}^{\rm H}\mathbf{z})=0$.
Consequently, the subspace $\mathcal{W}\triangleq\{\mathbf{z}\in\mathbb{C}^k: \mathbf{a}^{\rm H}\mathbf{z}=0\}$
corresponds to a real linear subspace of codimension $2$ in
$\mathbb{R}^{2k}$.
\end{proposition}

\begin{definition}[Absolute Continuity of Probability Measures~{\normalfont\cite[Section~9B]{AxlerMIRA2020}}]
\label{def:absolute-continuity}
A probability measure $P_{\boldsymbol{\theta}}$ on $\mathbb{R}^m$ is
absolutely continuous with respect to the Lebesgue measure $\lambda^m$
(written $P_{\mathbf{\Theta}}\ll\lambda^m$) if
$\lambda^m(E)=0 \implies P_{\mathbf{\Theta}}(E)=0$
for every measurable set $E\subset\mathbb{R}^m$.
\end{definition}

\begin{proposition}[Sampling Probability over Measure-Zero Sets~{\normalfont\cite[Section~9B]{AxlerMIRA2020}}]\label{prop:measure-zero-prob}
Let $\mathbf{\Theta}$ be a continuous random vector taking values in $\mathbb{R}^m$ whose probability distribution $P_{\mathbf{\Theta}}$ is absolutely continuous with respect to Lebesgue measure (such as complex circularly symmetric Gaussian or Ricean fading distributions). For any measurable set $E \subset \mathbb{R}^m$ with $\lambda^m(E) = 0$, the probability ${\rm Pr}(\mathbf{\Theta} \in E) = 0$.
\end{proposition}




\begin{proposition}[Polynomial Nature and Analyticity of Determinants~{\normalfont\cite[Section~9C]{AxlerLADR2024}}]\label{prop:determinant-analyticity}
The determinant of matrix $\det(\mathbf{A})$, $\det(\mathbf{A})$, is a multivariate polynomial in the entries $[\mathbf{A}]_{i,j}$, and is therefore real analytic on $\mathbb{R}^{M \times M}$.
\end{proposition}


\begin{definition}[Range / Column Space of a Matrix~{\normalfont\cite[Section~3B]{AxlerLADR2024}}]\label{def:matrix-range}
For any matrix $\mathbf{X} \in \mathbb{C}^{m \times k}$, its range (or column space) is defined as
$\mathcal{R}(\mathbf{X}) \triangleq \left\{ \mathbf{X} \mathbf{v} \in \mathbb{C}^m \;\middle|\; \mathbf{v} \in \mathbb{C}^k \right\} \subseteq \mathbb{C}^m$.
\end{definition}

\begin{proposition}[Properties of Pseudoinverse Projection Matrix~{\normalfont\cite[Section~6C]{AxlerLADR2024}}]\label{prop:pseudoinverse-properties}
For any matrix $\mathbf{G} \in \mathbb{C}^{D \times N_{\rm r}}$, the matrix product $\mathbf{G}\mathbf{G}^\dagger \in \mathbb{C}^{D \times D}$ is the unique orthogonal projection operator onto the range $\mathcal{R}(\mathbf{G})$, satisfying $(\mathbf{G}\mathbf{G}^\dagger)^{\rm H} = \mathbf{G}\mathbf{G}^\dagger$ and $(\mathbf{G}\mathbf{G}^\dagger)^2 = \mathbf{G}\mathbf{G}^\dagger$.
\end{proposition}

\begin{proposition}[Best Approximation Property of Orthogonal Projections~{\normalfont\cite[Section~6C]{AxlerLADR2024}}]\label{prop:best-approximation}
Let $\mathcal{V} \subseteq \mathbb{C}^D$ be a linear subspace and let $\mathbf{P}_{\mathcal{V}} \in \mathbb{C}^{D \times D}$ be the unique orthogonal projection matrix onto $\mathcal{V}$. For any vector $\mathbf{u} \in \mathbb{C}^D$, the projection $\mathbf{P}_{\mathcal{V}} \mathbf{u}$ minimizes the Euclidean distance from $\mathbf{u}$ to $\mathcal{V}$
$\|\mathbf{u} - \mathbf{P}_{\mathcal{V}} \mathbf{u}\|_2 = \min_{\mathbf{y} \in \mathcal{V}} \|\mathbf{u} - \mathbf{y}\|_2$.
\end{proposition}

\section{Proof of Proposition \ref{prop:universal-approximation}}\label{appendix:universal-approximation-proof}

To show that~\eqref{eq:LS-objective} admits a zero-error solution, it suffices to prove that $\det(\mathbf{G}) \neq 0$ with probability $1$. For this goal, we define the function $f_{\mathcal{D}}(\boldsymbol{\theta})\triangleq \det(\mathbf{G})$ on the domain $\mathcal{I} \subset \mathbb{R}^{2N_{\rm r} N_{\rm t} + N_{\rm r}}$, which is parametrized by the random coefficients of the channel and the biases (that affect the measurement matrix $\mathbf{G}$), for a given arbitrary data set $\mathcal{D}$. The definitions, propositions, and lemmas of the previous Appendix~\ref{appendix:definitions} are used. In particular, we make use of Proposition~\ref{prop:analytic-zero-set}
stating that, for a real analytic function on an open connected domain in which this function is not identically zero (i.e., $\mathcal{J} \triangleq \{ \boldsymbol{\theta} \in \mathcal{I}: f_{\mathcal{D}}(\boldsymbol{\theta}) \neq 0 \} \neq \emptyset$), the zero set $\mathcal{Z} \triangleq \mathcal{I} \setminus \mathcal{J}$ has zero Lebesgue measure
(see Definitions~\ref{def:real-analytic-function},~\ref{def:open-connected-domain},~\ref{def:zero-set} and~\ref{def:lebesgue-measure}), i.e., the set of random weights where $\det(\mathbf{G}) = 0$ is negligibly small compared to the full space of random parameters, irrespective of the chosen data set. 
Note that, despite $\boldsymbol{\theta}$ containing complex quantities, due to the isomorphism property between $\mathbb{C}$ and $\mathbb{R}^{2}$ (Definition~\ref{def:complex-real-isomorphism}), it suffices to examine $f_D: \mathcal{I} \to \mathbb{R}$ as a real function in $\mathbb{R}^{2N_{\rm r}N_{\rm t} + N_{\rm r}} \to \mathbb{R}$, thereby simplifying the rest of the derivations needed.

Since $f_{\mathcal{D}}(\boldsymbol{\theta})$ is a composition of real analytic functions, this function is indeed real analytic on $\mathcal{I}$ (Lemma~\ref{lem:analytic-composition}).
In fact, all of \eqref{eq:transmit-signal},~\eqref{eq:received-signal},~\eqref{eq:softplus-activation},~\eqref{eq:output} and $\det(\cdot)$ involve real analytic functions (Propositions~\ref{prop:elementary-analytic} and~\ref{prop:determinant-analyticity}).
Moving to the precise definition of $\mathcal{I}$, it needs to be defined to exclude points where the modulus inside~\eqref{eq:softplus-activation} equals zero, i.e., $ \mathcal{I} = \{ \mathbf{H}, \mathbf{b}:|[ \mathbf{H}\bar{\mathbf{x}}]_j| \neq 0, [\mathbf{b}]_j > 0\, \forall j=1,\dots,N_{\rm r} \}$ and must be connected. Each of the excluded subsets $\mathcal{I}^0_j \triangleq \{ [\mathbf{H}]_{j,:} : |[ \mathbf{H}\bar{\mathbf{x}}]_j| = 0 \}$ has an intrinsic dimension $N_{\rm r}-1$, and its elements therefore have zero probability of being sampled from the Ricean model; this happens since $\mathcal{I}^0_j$ has zero Lebesgue measure and the complex Gaussian distribution is absolutely continuous with respect to that measure (Definition~\ref{def:absolute-continuity} and Proposition~\ref{prop:measure-zero-prob}). Toward ensuring connectedness, first, the condition $[\mathbf{b}]_j>0$ inherently defines a convex, connected open region by Proposition~\ref{prop:convex-connected}.
Next, the excluded subsets $\mathcal{I}_j^{0}$ are defined by a complex linear constraint, which imposes two simultaneous real constraints (Proposition~\ref{prop:complex-codim}).
Consequently, each $\mathcal{I}_j^{0}$ constitutes a subspace of real codimension $2$ (Definition~\ref{def:real-codimension}) within the isomorphic real space. Since removing a subspace of real codimension $2$ or greater from a connected open set preserves its connectedness (Proposition~\ref{prop:codim-connectedness}), the domain $\mathcal{I}$ remains connected.

Finally, it remains to show that $\exists \boldsymbol{\theta} \in \mathcal{I}$ such that $f_{\mathcal{D}}(\boldsymbol{\theta}) \neq 0$.
Consider $\mathcal{I}^{\rm Re} \triangleq \{ \boldsymbol{\theta} \in \mathcal{I} : \mathfrak{Im}(\mathbf{H}) = \mathbf{0} \}$ and suppose without loss of generality that the selected phases ${\boldsymbol \psi}$ of the considered \gls{AM} data signals are set to $\mathbf{0}$.
For this case, the presented XL MIMO system is reduced to a standard real-valued \gls{ELM}, with a linear hidden layer and a continuous, non-polynomial activation function.
It therefore fulfills the conditions for universal interpolation~\cite[Theorem 2.2]{HuaZhuSie:06}), and consequently satisfies $\| \mathbf{z} - \mathbf{G} \mathbf{w} \|_2^{2} = 0$, implying $\det(\mathbf{G}) \neq 0$ on $\mathcal{I}^{\rm Re}$.

To sum up, $f_{\mathcal{D}}$ is a real analytic function, defined on an open and connected domain $\mathcal{I} = \mathbb{R}^{2N_{\rm r}N_{\rm t} + N_{\rm r}} \setminus \mathcal{I}^0 \triangleq \cup_{j=1}^{N_{\rm r}} \mathcal{I}_j^0 $, its zero set $\mathcal{I}^0$ has a zero Lebesgue measure, and in $\mathcal{I}^{\rm Re} \subset \mathcal{I}$, $f_{\mathcal{D}}(\boldsymbol{\theta}) \neq 0$ holds.
Then, by Proposition~\ref{prop:analytic-zero-set}, $f_{\mathcal{D}} \neq 0$ almost everywhere, therefore, matrix $\mathbf{G}$ in~\eqref{eq:LS-objective} is invertible, concluding the proof.

\section{Proof of Proposition \ref{prop:monotonicity}}\label{appendix:monotonicity-proof}

Let $1 \le M_1 \le M_2 \le N_{\rm r}$ denote nested sketch dimensions. We also define the measurement matrices $\mathbf{A}_{M_1}^{\rm H} \in \mathbb{C}^{N_{\rm r} \times M_1}$ and $\mathbf{A}_{M_2}^{\rm H} \in \mathbb{C}^{N_{\rm r} \times M_2}$, 
where the latters is formed by appending $M_2 - M_1$ additional measurement directions $\Delta\mathbf{A}^{\rm H} \in \mathbb{C}^{N_{\rm r} \times (M_2 - M_1)}$ to former 
such that $\mathbf{A}_{M_2}^{\rm H} = \begin{bmatrix} \mathbf{A}_{M_1}^{\rm H} & \Delta\mathbf{A}^{\rm H} \end{bmatrix}$. 
The corresponding sketched activation matrices are $\tilde{\mathbf{G}}_{M_1} \triangleq \mathbf{G} \mathbf{A}_{M_1}^{\rm H} \in \mathbb{C}^{D \times M_1}$ and $\tilde{\mathbf{G}}_{M_2} \triangleq \mathbf{G} \mathbf{A}_{M_2}^{\rm H} = \begin{bmatrix} \tilde{\mathbf{G}}_{M_1} & \mathbf{G} \Delta\mathbf{A}^{\rm H} \end{bmatrix} \in \mathbb{C}^{D \times M_2}$. For any vector $\mathbf{y}_1 \in \mathcal{R}(\tilde{\mathbf{G}}_{M_1})$ (Definition~\ref{def:matrix-range} in Appendix~\ref{appendix:definitions}), there exists another vector $\mathbf{v}_1 \in \mathbb{C}^{M_1 \times 1}$ such that $\mathbf{y}_1 = \tilde{\mathbf{G}}_{M_1} \mathbf{v}_1$. Constructing the augmented vector $\mathbf{v}_2 = \begin{bmatrix} \mathbf{v}_1^{\top} & \mathbf{0}^{\top} \end{bmatrix}^{\top} \in \mathbb{C}^{M_2 \times 1}$ yields:
\begin{equation*}
    \tilde{\mathbf{G}}_{M_2} \mathbf{v}_2 = \begin{bmatrix} \tilde{\mathbf{G}}_{M_1} & \mathbf{G} \Delta\mathbf{A}^{\rm H} \end{bmatrix} \begin{bmatrix} \mathbf{v}_1 \\ \mathbf{0} \end{bmatrix} = \tilde{\mathbf{G}}_{M_1} \mathbf{v}_1 = \mathbf{y}_1 \in \mathcal{R}(\tilde{\mathbf{G}}_{M_2}).
\end{equation*}
This establishes the following subspace containment chain:
\begin{equation}
    \mathcal{R}(\tilde{\mathbf{G}}_{M_1}) \subseteq \mathcal{R}(\tilde{\mathbf{G}}_{M_2}) \subseteq \mathcal{R}(\mathbf{G}).
\end{equation}

By Proposition~\ref{prop:pseudoinverse-properties} in Appendix~\ref{appendix:definitions}, the sketched prediction vector $\hat{\mathbf{z}}_{\text{sketch}}(M_k) = \tilde{\mathbf{G}}_{M_k} \tilde{\mathbf{G}}_{M_k}^\dagger \mathbf{z}$ from~\eqref{eq:sketch-output} represents the orthogonal projection of the target vector $\mathbf{z}$ onto $\mathcal{R}(\tilde{\mathbf{G}}_{M_k})$ for $k \in \{1, 2\}$.
Since $\mathcal{R}(\tilde{\mathbf{G}}_{M_k}) \subseteq \mathcal{R}(\mathbf{G})$ and the unsketched prediction vector $\hat{\mathbf{z}}_{\text{LS}} = \mathbf{G} \mathbf{G}^\dagger \mathbf{z}$ is the orthogonal projection of $\mathbf{z}$ onto $\mathcal{R}(\mathbf{G})$, projecting $\hat{\mathbf{z}}_{\text{LS}}$ onto $\mathcal{R}(\tilde{\mathbf{G}}_{M_k})$ yields $\hat{\mathbf{z}}_{\text{sketch}}(M_k)$.
Applying the best approximation property of orthogonal projections (Proposition~\ref{prop:best-approximation} in Appendix~\ref{appendix:definitions}), $\hat{\mathbf{z}}_{\text{sketch}}(M_k)$ minimizes the Euclidean distance from $\hat{\mathbf{z}}_{\text{LS}}$ to $\mathcal{R}(\tilde{\mathbf{G}}_{M_k})$, i.e.:
\begin{equation} \label{eq:M1_M2_min}
    \left\| \hat{\mathbf{z}}_{\text{LS}} - \hat{\mathbf{z}}_{\text{sketch}}(M_k) \right\|_2 = \min_{\mathbf{y} \in \mathcal{R}(\tilde{\mathbf{G}}_{M_k})} \left\| \hat{\mathbf{z}}_{\text{LS}} - \mathbf{y} \right\|_2.
\end{equation}

Owing to the subspace inclusion $\mathcal{R}(\tilde{\mathbf{G}}_{M_1}) \subseteq \mathcal{R}(\tilde{\mathbf{G}}_{M_2})$, minimizing over the superset $\mathcal{R}(\tilde{\mathbf{G}}_{M_2})$ cannot increase the achieved minimum distance, i.e., it holds:
\begin{equation}\label{eq:proj-min-distance-inneq}
    \min_{\mathbf{y} \in \mathcal{R}(\tilde{\mathbf{G}}_{M_2})} \left\| \hat{\mathbf{z}}_{\text{LS}} - \mathbf{y} \right\|_2 \le \min_{\mathbf{y} \in \mathcal{R}(\tilde{\mathbf{G}}_{M_1})} \left\| \hat{\mathbf{z}}_{\text{LS}} - \mathbf{y} \right\|_2.
\end{equation}
Substituting \eqref{eq:M1_M2_min} into \eqref{eq:proj-min-distance-inneq} directly yields:
\begin{equation}
    \left\| \hat{\mathbf{z}}_{\text{LS}} - \hat{\mathbf{z}}_{\text{sketch}}(M_2) \right\|_2 \le \left\| \hat{\mathbf{z}}_{\text{LS}} - \hat{\mathbf{z}}_{\text{sketch}}(M_1) \right\|_2,
\end{equation}
which completes the proof.

\section{Proof of Proposition \ref{prop:closed-form-MS}}\label{appendix:closed-form-MS-proof}
Let $Q(\boldsymbol{\phi}_L)\in\mathbb{C}^{1\times N_{\rm r}}$ represent the transfer function between the front-end \gls{NL} \gls{MS} layer, indexed with $l=0$, and the single antenna at the RX side of the proposed XL-MIMO-ELM system, which can be expressed as a function of the elements responses of the final $L$-indexed \gls{MS} layer of the linear \gls{SIM} as follows: 
\begin{equation} \label{eq:sim-response}
    Q(\boldsymbol{\phi}_L) \triangleq \boldsymbol{\omega} \mathbf{\Phi}_L \mathbf{\Omega}.
\end{equation}
Suppose also that the response of each $n$-th element ($n=1,\ldots,N_{\rm r}$) of the final \gls{MS} layer is constrained as $|\alpha^L_n| \le 1$ (phase-only, passive response). Following Section~\ref{sec:NL_SIM}, the output of the \gls{NL} \gls{MS} layer $g(\mathbf{y})$ 
propagates through the linear \gls{SIM}, as described in \eqref{eq:sim-response}, reaching the single RX antenna; this operation that can be mathematically expressed as $Q(\boldsymbol{\phi}_L)g(\mathbf{y})$.

Our objective is to realize an arbitrary linear functional of $g(\mathbf{y})$ via the \gls{SIM}, defined by the combining vector $\mathbf{w}$, such that:
\begin{equation}\label{eq:z_target}
    z_{\text{target}} \triangleq \sum_{n=1}^{N_{\rm r}} [\mathbf{w}]^\ast_n [g(\mathbf{y})]_n = \mathbf{w}^{\rm H} g(\mathbf{y}).
\end{equation}
Allowing for an overall real-valued scaling factor $\gamma > 0$, we specifically seek the conditions under which $Q(\boldsymbol{\phi}_L) = \gamma \mathbf{w}^{\rm H}$. Expanding this equality using~\eqref{eq:sim-response} yields:
\begin{align} \label{eq:linear_function1}
    \boldsymbol{\omega} \mathbf{\Phi}_L \mathbf{\Omega} = \gamma \mathbf{w}^{\rm H} \!\!\!\implies\!\!\! \boldsymbol{\omega} \mathbf{\Phi}_L \mathbf{\Omega} g(\mathbf{y})
    \!\overset{(b)}{=} \! \sum_{n=1}^{N_{\rm r}} [\boldsymbol{\omega}]_n [\boldsymbol{\phi}_L]_n [\mathbf{s}]_n,
\end{align}
where (b) uses the definition $\mathbf{s} \triangleq \mathbf{\Omega} g(\mathbf{y}) \in \mathbb{C}^{N_{\rm r} \times 1}$. Since $\mathbf{\Omega}$ is an invertible full-rank matrix (Proposition~\ref{prop:sim-full-rank} and its proof in the following Appendix~\ref{appendix:sim-full-rank-proof}), there exists an one-to-one correspondence between $g(\mathbf{y})$ and $\mathbf{s}$. The target functional in~\eqref{eq:z_target} can thus be equivalently expressed in the $\mathbf{s}$-coordinate domain as follows:
\begin{equation} \label{eq:linear_target}
    z_{\text{target}} = \mathbf{w}^{\rm H} \mathbf{\Omega}^{-1} \mathbf{s} \overset{(c)}{=} \sum_{n=1}^{N_{\rm r}} [\tilde{\mathbf{w}}]^\ast_n [\mathbf{s}]_n,
\end{equation}
where (c) makes use of the definition $\tilde{\mathbf{w}}^{\rm H}\triangleq \mathbf{w}^{\rm H} \mathbf{\Omega}^{-1}$ with $\tilde{\mathbf{w}}^{\rm H} \in \mathbb{C}^{1 \times N_{\rm r}}$ acting as the transformed target weight vector. Equating the term-by-term summations of the physical system and the target functional in the $\mathbf{s}$-domain, we require the optimal configuration of the final $L$-indexed \gls{MS} layer, $\boldsymbol{\phi}_L^{\mathbf{w}}$, to satisfy:
\begin{equation}
    [\boldsymbol{\omega}]_n [\boldsymbol{\phi}_L^{\mathbf{w}}]_n = \gamma [\tilde{\mathbf{w}}^{\rm H}]_n\,\,\forall n=1, \dots, N_{\rm r}.
\end{equation}
Since the physical channel from the \gls{SIM} to the single antenna contains no perfect nulls (i.e, $[\boldsymbol{\omega}]_n \neq 0$~$\forall n$), these scalar equations can be solved independently for each element, yielding:
\begin{equation}
    [\boldsymbol{\phi}_L^{\mathbf{w}}]_n = \frac{\gamma [\tilde{\mathbf{w}}]_n}{[\boldsymbol{\omega}]_n} \overset{(c)}{=} \frac{\gamma}{[\boldsymbol{\omega}]_n} \left[ \mathbf{w}^{\rm H} \mathbf{\Omega}^{-1} \right]_n.
\end{equation}

Note that to ensure the solution is physically realizable with passive elements, the constraint $|[\boldsymbol{\phi}_L^*]_n| \le 1$ must hold $\forall n$. Let $\zeta \triangleq \max_n \left( \left| \frac{1}{[\boldsymbol{\omega}]_n} \left[ \mathbf{w}^{\rm H} \mathbf{\Omega}^{-1} \right]_n \right| \right)$. The scaling factor $\gamma$ must, therefore, be chosen such that $\gamma \le \frac{1}{\zeta}$. Under this condition, for any desired coefficient vector $\mathbf{w}^{\rm H}$, the phase settings of the final $L$-indexed linear MS layer at the proposed RX architecture are given analytically by the configuration stated in the proposition, confirming that the physical layer can exactly compute the desired ELM functional in the wave domain, scaled by $\gamma$, successfully bridging the theoretical target with the actual system output given by~\eqref{eq:output}.

\section{Proof of Proposition \ref{prop:sim-full-rank}}\label{appendix:sim-full-rank-proof}
The $N_{\rm r} \times N_{\rm r}$ propagation matrix $\mathbf{\Omega}$ is full rank if and only if its determinant is non-zero. Using the multiplicativity property of the determinant, $\mathbf{\Omega}$'s determinant can be expressed as follows:
\begin{equation}
    \det(\mathbf{\Omega}) = \det(\mathbf{\Psi}_L) \prod_{l=1}^{L-1} \det(\mathbf{\Phi}_l) \det(\mathbf{\Psi}_l).
\end{equation}
Consequently, $\mathbf{\Omega}$ is non-singular if and only if $\det(\mathbf{\Phi}_l) \neq 0$ and $\det(\mathbf{\Psi}_l) \neq 0$ $\forall$~$l=1,\dots,L-1$. Let us first consider the configuration matrix $\mathbf{\Phi}_l$ at the $l$-th \gls{MS} layer, defined as:
\begin{equation}
    \mathbf{\Phi}_l\!=\!\operatorname{diag}\left(\boldsymbol{\phi}_l\right)\!=\! 
    \operatorname{diag}\left(\left[|\alpha_1^l|e^{-j\vartheta_1^l}, \dots, |\alpha_{N_{\rm r}}^l|e^{-j\vartheta_{N_{\rm r}}^l}\right]\right)\!.
\end{equation}
Its determinant is easily computed as follows:
\begin{equation}
    \det(\mathbf{\Phi}_l) = \prod_{n=1}^{N_{\rm r}} |\alpha_n^l| e^{-j\vartheta_n^l}.
\end{equation}
Given that the elements responses $\{\boldsymbol{\phi}_l\}_{l=1}^{L-1}$ are arbitrarily sampled from an absolutely continuous distribution whose support strictly excludes zero for the amplitudes ($|\alpha_n^l| > 0$), it follows that $|\det(\mathbf{\Phi}_l)| > 0$ almost surely.

Next, let $d_M > 0$ denote the strictly positive inter-\gls{MS}-layer distance, $\mathbf{r}_{n'}^{(l-1)} \triangleq [x'_{n'}, y'_{n'}]^\top$ the fixed spatial positions on layer $l-1$, and $\mathbf{p} \triangleq [x_1, y_1, \dots, x_{N_{\rm r}}, y_{N_{\rm r}}]^\top \in \mathbb{R}^{2N_{\rm r}}$ include the coordinates of all elements on layer $l$. The Euclidean distance between element $n'$ on layer $l-1$ and element $n$ on layer $l$ is:
\begin{equation}
    d_{n,n'}(\mathbf{p}) \triangleq \sqrt{\|\mathbf{r}_n^{(l)} - \mathbf{r}_{n'}^{(l-1)}\|^2 + d_M^2}.
\end{equation}
which can be used in~\eqref{eq:rayleigh-somerfeld} to explicitly express the propagation matrix at this layer as a function of $\mathbf{p}$, $[\mathbf{\Psi}_l]_{n,n'}(\mathbf{p})$. Note that $d_{n,n'}(\mathbf{p})$ is a real-analytic function of $\mathbf{p}$ on the open, connected domain $\mathbb{R}^{2N_{\rm r}}$, and that, since $d_M > 0$, the radicand is strictly positive $\forall \mathbf{p} \in \mathbb{R}^{2N_{\rm r}}$. 
Furthermore, using Lemma~\ref{lem:analytic-composition} and Proposition~\ref{prop:elementary-analytic}, the composition of this distance function with the rational, algebraic, and complex exponential terms preserves analyticity. Therefore, the entry $[\mathbf{\Psi}_l]_{n,n'}(\mathbf{p})$ is real-analytic on $\mathbb{R}^{2N_{\rm r}}$. By Proposition~\ref{prop:determinant-analyticity}, the determinant $h(\mathbf{p}) \triangleq \det(\mathbf{\Psi}_l(\mathbf{p}))$ is a multivariate polynomial in its entries, implying that $h: \mathbb{R}^{2N_{\rm r}} \to \mathbb{C}$ is real-analytic on $\mathbb{R}^{2N_{\rm r}}$.

To verify that $h(\mathbf{p})$ is not identically zero, we analyze the asymptotic behavior of a parametrized coordinate family $\mathbf{p}(\delta)$ as the \gls{MS} element's pitch $\delta \to \infty$. To this end, let $\mathbf{r}_n^{(l)}(\delta) \triangleq \mathbf{r}_n^{(l-1)} + \delta \mathbf{v}_n$, where $\mathbf{v}_n$ is a distinct two-dimensional unit directional vector. For $\mathbf{\Psi}_l$'s diagonal entries where $n' = n$, we have $\mathbf{v}_n - \mathbf{v}_{n'} = 0$, which yields purely longitudinal propagation such that $d_{n,n}(\mathbf{p}(\delta)) = d_M$. Consequently, the diagonal entries converge to a non-zero constant $C$, which is given by:
\begin{equation}
    C\triangleq [\mathbf{\Psi}_l]_{n,n} = \frac{S_M}{d_M} \left(\frac{1}{2\pi d_M} - \frac{\jmath}{\lambda}\right) e^{\frac{\jmath2\pi}{\lambda} d_M} \neq 0.
\end{equation}
Conversely, for $\mathbf{\Psi}_l$'s off-diagonal entries where $n' \neq n$, the distance scales as $d_{n,n'}(\mathbf{p}(\delta)) \ge \delta \|\mathbf{v}_n - \mathbf{v}_{n'}\|$, which approaches infinity as $\delta \to \infty$. Thus, $\lim_{\delta \to \infty} [\mathbf{\Psi}_l]_{n,n'}(\delta) = 0$ $\forall n' \neq n$. Hence, evaluating the limit of the determinant yields:
\begin{equation}
    \lim_{\delta \to \infty} \det(\mathbf{\Psi}_l(\mathbf{p}(\delta))) = \det\left( C \mathbf{I}_{N_{\rm r}} \right) = C^{N_{\rm r}} \neq 0.
\end{equation}
Since the limit is non-zero, the continuous function $h(\mathbf{p})$ cannot be identically zero on $\mathbb{R}^{2N_{\rm r}}$. 

Having established that $h(\mathbf{p})$ is real analytic on the open connected domain $\mathbb{R}^{2N_{\rm r}}$ and that is not identically zero, we invoke Proposition~\ref{prop:analytic-zero-set}. The zero set, defined as $\mathcal{Z} = \{\mathbf{p} \in \mathbb{R}^{2N_{\rm r}} \mid \det(\mathbf{\Psi}_l(\mathbf{p})) = 0\}$, has Lebesgue measure zero. Consequently, for almost all physical configurations of the spatial coordinates, $\det(\mathbf{\Psi}_l) \neq 0$ holds. Combining this with the almost sure non-singularity of $\mathbf{\Phi}_l$, it follows that $\det(\mathbf{\Omega}) \neq 0$ almost surely, proving that $\operatorname{rank}(\mathbf{\Omega}) = N_{\rm r}$.

\end{appendices}

\balance

\bibliographystyle{IEEEtran}
\bibliography{references}

@ARTICLE{Abb:26,
  author={Abbas, Omran and others},
  journal={IEEE Wireless Commun.}, 
  title={Nonlinear Stacked Intelligent Surfaces for Wireless Systems}, 
  year={early access, 2026}
  }

@ARTICLE{9475155,
  author={Alexandropoulos, George C. and others},
  journal={IEEE Commun. Mag.}, 
  title={Reconfigurable Intelligent Surfaces for Rich Scattering Wireless Communications: Recent Experiments, Challenges, and Opportunities}, 
  year={2021},
  volume={59},
  number={6},
  pages={28-34}
  }

@book{AlexandropoulosRIS,
  author    = {Alexandropoulos, G. C. and others},
  title     = {Reconfigurable Intelligent Surfaces for Wireless Communications: Modeling, Architectures, and Applications},
  publisher = {Springer Singapore},
  year      = {2026}
}

@ARTICLE{Nin:25,
  author={Ning, Yu Ming and others},
  journal={Nat. Commun.}, 
  title={Multilayer nonlinear diffraction neural networks with programmable and fast {ReLU} activation function}, 
  year={2025},
  volume={16},
  number={10332},
  pages={1-13},
  doi={https://doi.org/10.1038/s41467-025-65275-0}
  }

@INPROCEEDINGS{Abr:26,
  author={Abrardo, Andrea and Toccafondi, Alberto},
  booktitle={Proc. IEEE Int. Mediterranean Conf. Commun. Netw.}, 
  title={Stacked Intelligent Metasurfaces {(SIM)} in the Nonlinear Regime: A Multiport Network Model Approach}, 
  year={2026},
  address={Cagliari, Italy}
  }

@ARTICLE{DarTorPasDec:J26,
  author={Dardari, Davide and others},
  journal={IEEE Signal Process. Mag.}, 
  title={Over-the-Air Electromagnetic Signal Processing: From fundamentals to emerging paradigms}, 
  year={2026},
  volume={43},
  number={1},
  pages={6-28},
  doi={10.1109/MSP.2025.3622001},
  ISSN={1558-0792}
  }

@book{Maas2003,
  author    = {Stephen A. Maas},
  title     = {Nonlinear Microwave and RF Circuits},
  edition   = {2nd},
  publisher = {Artech House},
  address   = {Boston, MA, USA},
  year      = {2003}
}

@INPROCEEDINGS{FabTorDar:25,
  author={M. Fabiani and others},
  booktitle={Proc.  Asilomar Conf. Signals, Sys., and Comput.}, 
  title={Nonlinear {EM-based} Signal Processing}, 
  year={2025},
  address={Pacific Grove, USA}
}

@article{Styl:25_Mag,
  title={Metasurfaces-integrated wireless neural networks for lightweight over-the-air edge inference},
  author={K. Stylianopoulos and others},
  journal={IEEE Wireless Commun.},
  volume={3},
  number={4},
  pages={82-90},
  year={2026},
}

@article{Omr:25,
  title={Nonlinear Stacked Intelligent Surfaces for Wireless Systems},
    author={Abbas, Omran and others},
  journal={arXiv preprint arXiv:2510.23780},
  year={2025}
  }

@article{SIM_survey,
  title={A survey on stacked intelligent metasurfaces: Fundamentals, recent advances, and challenges},
    author={C. K. Sheemar and others},
  journal={arXiv preprint arXiv:2603.05633},
  year={2026}
  }

@article{Cyb:89,
  author = {G. Cybenko},
title = {Approximation by superpositions of a sigmoidal function},
journal = {Math. Control Signal Sys.},
volume = {2},
pages = {303–314},
year = {1989}
}

@article{HuaLiCheSie:08,
  author = {Guang-Bin Huang and others},
title = {Incremental extreme learning machine with fully complex hidden nodes},
journal = {Neurocomput.},
volume = {71},
number = {4},
pages = {576-583},
year = {2008},
issn = {0925-2312},
}

@article{LiHuaSarSun:05,
 author = {Ming-Bin Li and others},
title = {Fully complex extreme learning machine},
journal = {Neurocomput.},
volume = {68},
pages = {306-314},
year = {2005},
}

@article{HuaZhuSie:06,
author = {Guang-Bin Huang and others},
title = {Extreme learning machine: {T}heory and applications},
journal = {Neurocomput.},
volume = {70},
number = {1},
pages = {489-501},
year = {2006}
}

@article{Stylianopoulos_MIMO_ELM,
    author = {Stylianopoulos, K. and Alexandropoulos, G. C.},
    title = {Universal Approximation with {XL} {MIMO} Systems: {OTA} Classification via Trainable Analog Combining},
    journal = {arXiv preprint:2504.12758},
    year = {2025}
}

@INPROCEEDINGS{Rapp_model,
       author = {{Rapp}, Christoph},
        title = "{Effects of {HPA}-nonlinearity on a {4-DPSK}/{OFDM}-signal for a digital sound broadcasting signal}",
    booktitle = {ESA Special Publications Series},
         year = 1991,
       editor = {{Kaldeich}, Brigitte},
       volume = {332},
        pages = {179-184},
}

@article{XYN18_Diffractive_DNN,
author = {Xing Lin and others},
title = {All-optical machine learning using diffractive deep neural networks},
journal = {Science},
volume = {361},
number = {6406},
pages = {1004-1008},
year = {2018},
}

@ARTICLE{Stylianopoulos_GO,
  author={Stylianopoulos, Kyriakos and others},
  journal={IEEE Trans. Wireless Commun.}, 
  title={Over-the-Air Edge Inference via End-to-End Metasurfaces-Integrated Artificial Neural Networks}, 
  year={2026},
  volume={25},
  number={},
  pages={13818-13834}
  }

@ARTICLE{AXN23_SIM,
  author={An, Jiancheng and others},
  journal={IEEE J. Sel. Areas Commun.}, 
  title={Stacked Intelligent Metasurfaces for Efficient Holographic {MIMO} Communications in {6G}}, 
  year={2023},
  volume={41},
  number={8},
  pages={2380-2396},
}

@misc{breast_cancer_dataset,
  author={William Nick Street and others},
  title        = {Diagnostic {Wisconsin} Breast Cancer Database},
  howpublished = {{UC} {Irvine} Machine Learning Repository},
  note         = {doi: 10.24432/C5DW2B},
  year={2008}
}

@article{MNIST,
  title={The {MNIST} database of handwritten digit images for machine learning research},
  author={Deng, Li},
  journal={IEEE Signal Process. Mag.},
  volume={29},
  number={6},
  pages={141--142},
  year={2012},
  publisher={IEEE}
}

@ARTICLE{Parkinsons_dataset,
  author={M. A. Little and others},
  journal={IEEE Trans. Biomed. Eng.}, 
  title={Suitability of Dysphonia Measurements for Telemonitoring of {P}arkinson's Disease}, 
  year={2009},
  volume={56},
  number={4},
  pages={1015-1022},
}

@article{GJZ24_SIM_TOC,
      title={Stacked Intelligent Metasurfaces for Task-Oriented Semantic Communications}, 
      author={Guojun Huang and others},
      journal={IEEE Wireless Commun. Lett.},
      year={2025},
      volume={14},
      number={2},
      pages={2310-314}
}

@Article{Momeni2022_wave_DL_computing,
author={Momeni, Ali
and Fleury, Romain},
title={Electromagnetic wave-based extreme deep learning with nonlinear time-Floquet entanglement},
journal={Nature Commun.},
year={2022},
month={May},
day={12},
volume={13},
number={1},
pages={2651},
}

@ARTICLE{6184250,
  author={Matthaiou, Michail and others},
  journal={IEEE Trans. Commun.}, 
  title={Analytic Framework for the Effective Rate of {MISO} Fading Channels}, 
  year={2012},
  volume={60},
  number={6},
  pages={1741-1751}
}

@ARTICLE{GO_review2,
  author={Li, Aimin and others},
  journal={IEEE Wireless Commun.}, 
  title={Toward Goal-Oriented Semantic Communications: {N}ew Metrics, Framework, and Open Challenges}, 
  year={2024},
  volume={31},
  number={5},
  pages={238-245},
  }

@ARTICLE{DeepSC,
  author={Xie, Huiqiang and Qin, Zhijin and Li, Geoffrey Ye and Juang, Biing-Hwang},
  journal={IEEE Trans. Signal Process.}, 
  title={Deep Learning Enabled Semantic Communication Systems}, 
  year={2021},
  volume={69},
  number={},
  pages={2663-2675},
}

@ARTICLE{OTA_review,
  author={Şahin, Alphan and Yang, Rui},
  journal={IEEE Commun. Surveys \& Tuts.}, 
  title={A Survey on Over-the-Air Computation}, 
  year={2023},
  volume={25},
  number={3},
  pages={1877-1908},
}

@article{Gunduz_Layer_Approximation,
      title={Implementing Neural Networks Over-the-Air via Reconfigurable Intelligent Surfaces}, 
      author={Meng Hua and others},
      journal={IEEE Trans. Wireless Commun.},
      year={2026},
      volume={25},
      pages={11562-11576}
}

@INPROCEEDINGS{Pandolfo_SIM_semantic_alignment,
    author = {M. E. Pandolfo and others},
    title = {Over-the-air semantic alignment with stacked intelligent metasurfaces},
    booktitle = {Proc. European Signal Process. Conf.},
    address={Bruges, Belgium},
    year = {2026}
}

@INPROCEEDINGS{AYG24_SIM_DFT,
  author={An, Jiancheng and others},
  booktitle={Proc. IEEE Int. Commun. Conf.}, 
  title={Stacked Intelligent Metasurface Performs a {2D} {DFT} in the Wave Domain for {DOA} Estimation}, 
  year={2024},
  address={Denver, USA},
}

@Article{Mityagin2020,
author={Mityagin, B. S.},
title={The Zero Set of a Real Analytic Function},
journal={Math. Notes},
year={2020},
month={Mar},
day={01},
volume={107},
number={3},
pages={529-530},
issn={1573-8876},
}

@Article{Wiemann2024_softplus,
author={Wiemann, Paul F. V. and others},
title={Using the softplus function to construct alternative link functions in generalized linear models and beyond},
journal={Stat. Papers},
year={2024},
volume={65},
number={5},
pages={3155-3180}
}

@inproceedings{reus2023airfc,
  title={Air{FC}: Designing fully connected layers for neural networks with wireless signals},
  author={Reus-Muns, Guillem and others},
  booktitle={Proc. ACM MobiHoc},
  pages={71--80},
  year={2023}
}

@article{yang2023over,
  title={Over-the-air split machine learning in wireless {MIMO} networks},
  author={Yang, Yuzhi and others},
  volume={41},
  number={4},
  pages={1007--1022},
  journal={IEEE J. Sel. Areas Commun.},
  year={2022}
}

@article{sanchez2022airnn,
  title={Air{NN}: {O}ver-the-air computation for neural networks via reconfigurable intelligent surfaces},
  author={Sanchez, Sara Garcia and others},
  journal={IEEE Trans. Netw.},
  volume={31},
  number={6},
  pages={2470--2482},
  year={2022},
  publisher={IEEE}
}

@INPROCEEDINGS{zhang2024radio,
  author={Zhang, Jingyuan and others},
  booktitle={Proc. IEEE Veh. Technol. Conf.}, 
  title={A radio-frequency-based 2-{D} convolutional layer using transmissive intelligent surfaces},
  year={2024},
  address={Washington, DC, USA},
}

@article{hua2026cnns,
  title={{CNNs} in the air via reconfigurable intelligent surfaces},
  author={Hua, Meng and others},
  journal={IEEE Wireless Commun. Lett.},
  volume={15}, 
  pages={2124--2128},
  year={2026},
}

@article{liu2025over,
  title={Over-the-air {ODE}-inspired neural network for dual task-oriented semantic communications},
  author={Liu, Mengbing and others},
  journal={IEEE Trans. Cogn. Commun. Netw.}, 
  volume={12},
  pages={805--819},
  year={2025},
}

@article{hammami2026expressivity,
  title={Expressivity of Programmable-Metasurface-Based Physical Neural Networks: {E}ncoding Non-Linearity, Structural Non-Linearity, and Depth},
  author={Cheima Hammami and others},
  journal={arXiv preprint arXiv:2603.13602},
  year={2026}
}

@article{bergel2024non,
  title={Non-linear relay optimization using deep-learning tools},
  author={Bergel, Itsik},
  journal={IEEE Commun. Mag.},
  volume={23},
  number={12},
  pages={19289--19301},
  year={2024},
  publisher={IEEE}
}

@article{hua2026multi,
  title={Multi-layer {MIMO} Relay as Deep Physical Neural Networks: Power Amplifiers as Activation Functions},
  author={Meng Hua and others},
  journal={arXiv preprint arXiv:2607.18354},
  year={2026}
}

@book{KrantzParks2002,
  author    = {Steven G. Krantz and Harold R. Parks},
  title     = {A Primer of Real Analytic Functions},
  edition   = {2},
  series    = {Birkh{\"a}user Advanced Texts},
  publisher = {Birkh{\"a}user},
  address   = {Boston, MA},
  year      = {2002},
  doi       = {10.1007/978-0-8176-8134-0}
}

@book{Munkres2000,
  author    = {James R. Munkres},
  title     = {Topology},
  edition   = {2},
  publisher = {Prentice Hall},
  address   = {Upper Saddle River, NJ},
  year      = {2000},
  isbn      = {978-0-13-181629-9}
}

@book{AxlerLADR2024,
  author    = {Sheldon Axler},
  title     = {Linear Algebra Done Right},
  edition   = {4},
  series    = {Undergraduate Texts in Mathematics},
  publisher = {Springer},
  address   = {Cham, Switzerland},
  year      = {2024},
  doi       = {10.1007/978-3-031-41026-0}
}

@book{AxlerMIRA2020,
  author    = {Sheldon Axler},
  title     = {Measure, Integration \& Real Analysis},
  series    = {Graduate Texts in Mathematics},
  volume    = {282},
  publisher = {Springer},
  address   = {Cham, Switzerland},
  year      = {2020},
  doi       = {10.1007/978-3-030-33143-6}
}

@misc{ConradHyperplaneComplements,
  author       = {Brian Conrad},
  title        = {Connectedness of Hyperplane Complements},
  howpublished = {Math 396 course handout, Stanford University},
  note         = {Available online},
  url          = {https://math.stanford.edu/~conrad/diffgeomPage/handouts/hyperplaneconnd.pdf},
  urldate      = {2026-09-30}
}

@article{Qian2025Nanosecond,
  author  = {Qian, Boyang and others},
  title   = {Nanosecond-Level Time-Domain Coding Metasurface for Radar Signal Generation},
  journal = {Nat. Commun.},
  year    = {2025},
  volume  = {16},
  number  = {1},
  pages   = {10651},
}

@article{Li2017Electromagnetic,
  author  = {Li, Lianlin and others},
  title   = {Electromagnetic Reprogrammable Coding-Metasurface Holograms},
  journal = {Nat. Commun.},
  year    = {2017},
  volume  = {8},
  number  = {1},
  pages   = {197},
}

@INPROCEEDINGS{6G-DISAC-Architecture,
  author={Costanzo, Francesca and others},
  booktitle={Proc. Joint European Conf. Netw. Commun. \& 6G Summit}, 
  title={{DISAC} Architecture: {F}unctions, Semantics, Protocols, Orchestration, and Demonstrator}, 
  year={2026},
  address = {Malaga, Spain}
}

@inproceedings{Our_SPAWC_paper,
  author = {K. Stylianopoulos and others},
  title = {Over-the-Air Extreme Learning Machines with Nonlinear Stacked Intelligent Metasurfaces},
  booktitle = {Proc. IEEE Int. Workshop Signal Process. Adv. Wireless Commun.},
  year = {2026},
  address = {Athens, Greece}
}

@incollection{Alexandropoulos2026ReceivingRIS,
  author    = {Alexandropoulos, George C. and others},
  title     = {Receiving {RISs} for Channel Estimation and Autonomous Configuration},
  booktitle = {Reconfigurable Intelligent Surfaces for Wireless Communications: Modeling, Architectures, and Applications},
  editor    = {Alexandropoulos, George C. and others},
  publisher = {Springer Singapore},
  year      = {2026}
}

@article{Yuan2023DFT,
  author  = {Yuan, Minghao and others},
  title   = {Alternating Optimization Based Hybrid Transceiver Designs for Wideband Millimeter-Wave Massive Multiuser {MIMO-OFDM} Systems},
  journal = {IEEE Trans. Wireless Commun.},
  year    = {2023},
  volume  = {22},
  number  = {12},
  pages   = {9201--9217},
  doi     = {10.1109/TWC.2023.3269056}
}

\end{document}